\pdfoutput=1
\RequirePackage{ifpdf}
\ifpdf 
\documentclass[pdftex]{sigma}
\else
\documentclass{sigma}
\fi

\usepackage{dsfont,
mathdots,
arydshln,
blkarray,
tikz}

\newcommand\Cb{\mathbb{C}}
\newcommand\Zb{\mathbb{Z}}
\newcommand\Pb{\mathbb{P}}
\newcommand{\cN}{{\mathcal{N}}}
\newcommand{\bs}{\boldsymbol}
\newcommand{\la}{\langle} \newcommand{\ra}{\rangle}
\def\mf{\mathfrak}

\numberwithin{equation}{section}

\newtheorem{Theorem}{Theorem}[section]

\newtheorem{Lemma}[Theorem]{Lemma}
\newtheorem{Proposition}[Theorem]{Proposition}
 { \theoremstyle{definition}

\newtheorem{Remark}[Theorem]{Remark}
 }

\begin{document}
\allowdisplaybreaks

\newcommand{\arXivNumber}{1806.08650}

\renewcommand{\thefootnote}{}

\renewcommand{\PaperNumber}{123}

\FirstPageHeading

\ShortArticleName{On Solutions of the Fuji--Suzuki--Tsuda System}

\ArticleName{On Solutions of the Fuji--Suzuki--Tsuda System\footnote{This paper is a~contribution to the Special Issue on Painlev\'e Equations and Applications in Memory of Andrei Kapaev. The full collection is available at \href{https://www.emis.de/journals/SIGMA/Kapaev.html}{https://www.emis.de/journals/SIGMA/Kapaev.html}}}

\Author{Pavlo GAVRYLENKO~$^{\dag^1\dag^2\dag^3}$, Nikolai IORGOV~$^{\dag^1\dag^4}$ and Oleg LISOVYY~$^{\dag^5}$}

\AuthorNameForHeading{P.~Gavrylenko, N.~Iorgov and O.~Lisovyy}

\Address{$^{\dag^1}$~Bogolyubov Institute for Theoretical Physics, 03143 Kyiv, Ukraine}
\EmailDD{\href{mailto:pasha145@gmail.com}{pasha145@gmail.com}, \href{mailto:iorgov@bitp.kiev.ua}{iorgov@bitp.kiev.ua}}

\Address{$^{\dag^2}$~Center for Advanced Studies, Skolkovo Institute of Science and Technology,\\
\hphantom{$^{\dag^2}$}~143026 Moscow, Russia}

\Address{$^{\dag^3}$~National Research University Higher School of Economics, International Laboratory\\
\hphantom{$^{\dag^3}$}~of Representation Theory and Mathematical Physics, Moscow, Russia}

\Address{$^{\dag^4}$~Kyiv Academic University, 36 Vernadsky Ave., 03142 Kyiv, Ukraine}

\Address{$^{\dag^5}$~Institut Denis-Poisson, Universit\'e de Tours, Parc de Grandmont, 37200 Tours, France}
\EmailDD{\href{mailto:lisovyi@lmpt.univ-tours.fr}{lisovyi@lmpt.univ-tours.fr}}

\ArticleDates{Received June 22, 2018, in final form October 30, 2018; Published online November 11, 2018}

\Abstract{We derive Fredholm determinant and series representation of the tau function of the Fuji--Suzuki--Tsuda system and its multivariate extension, thereby generalizing to higher rank the results obtained for Painlev\'e~VI and the Garnier system. A special case of our construction gives a higher rank analog of the continuous hypergeometric kernel of Borodin and Olshanski. We also initiate the study of algebraic braid group dynamics of semi-degenerate monodromy, and obtain as a byproduct a direct isomonodromic proof of the AGT-W relation for ${c=N-1}$.}

\Keywords{isomonodromic deformations; Painlev\'e equations; Fredholm determinants}

\Classification{33E17; 34M55; 34M56}

\renewcommand{\thefootnote}{\arabic{footnote}}
\setcounter{footnote}{0}

\section{Introduction}

The purpose of this note is to initiate a systematic study of rank $N$ Fuji--Suzuki--Tsuda system, abbreviated below as FST$_N$. This Hamiltonian system of nonlinear non-autonomous ODEs first appeared as a particular reduction of the Drinfeld--Sokolov hierarchy \cite{FS,Suzuki1}, and independently in \cite{Tsuda0} as a reduction of the universal character hierarchy. Its fundamental significance comes from the isomonodromic theory~\cite{FIKN}, where it describes deformations of rank $N$ Fuchsian systems with 4 regular singular points, 2 of which have special spectral type $(N-1,1)$ \cite{Tsuda0,Tsuda}. Following~\cite{GIL18}, we refer to such linear systems as semi-degenerate.

For $N=2$, the spectral profile remains unconstrained and the corresponding FST$_2$ system is equivalent to the sixth Painlev\'e equation (PVI). For general $N$, the dimension of the phase space of FST$_N$ is $2(N-1)$, which is to be compared with the dimension $2(N-1)^2$ of the system of deformation equations for generic 4-point Fuchsian case. The FST$_N$ system is thus the closest relative of PVI in higher rank. It generalizes PVI in the direction different from the much studied Garnier system, which corresponds to increasing the number of singular points while keeping fixed the rank $N=2$ of the associated linear problem.
Loosely speaking, going from PVI to FST$_N$ is a nonlinear counterpart of the generalization of the Gauss $_2F_1$ to Clausen--Thomae $_N F_{N-1}$ hypergeometric function; the Garnier system would correspond to multivariate functions of Lauricella type. This analogy is exhibited already at the level of special function solutions but in fact it goes much further: we will see that the general Fredholm determinant solution of FST$_N$ can be constructed from the fundamental solutions of two auxiliary linear $_N F_{N-1}$-systems.

Besides serving as a model example for isomonodromic deformations in higher rank, the FST$_N$ system appears in a number of applications. Its tau function is given by a Fourier transform of 4-point semi-degenerate conformal blocks of the Toda CFT with central charge \smash{$c=N-1$}~\cite{GIL18}. By the AGT-W correspondence \cite{AGT,FL3,MM,Wyl}, it thus coincides with the dual instanton partition of $\mathcal N=2$ $\mathrm{U}(N)$ gauge theory with $N$ fundamental and $N$ anti-fundamental matter hypermultiplets in the self-dual $\Omega$-background. Quantized FST$_N$ system also appears in this context~\cite{Yamada}.

On the other hand, little is known about the actual solutions of FST$_N$. The present paper begins to explore an agenda of related issues, to a large extent determined by the Painlev\'e~VI state-of-the-art. Our main tool is the Riemann--Hilbert correspondence which translates the questions on solutions of the nonlinear isomonodromic system in terms of monodromy of the associated linear problem. For example, all known PVI solutions expressible in terms of elementary or classical special functions can be divided into 3 (overlapping) categories:
\begin{itemize}\itemsep=0pt
\item \textit{Riccati solutions}. They correspond to monodromy representations generated from the reducible ones by Painlev\'e VI affine symmetry transformations \cite{Okamoto1}. The Riccati tau functions are expressed as finite determinants with hypergeometric $_2F_1$-entries. Their FST$_N$ analogs were investigated in \cite{ManoTsuda, Suzuki2,Tsuda2} and similarly involve $_NF_{N-1}$.
\item \textit{Picard solutions}. These solutions \cite{DIKZ,KK,Mazzocco} are expressed in terms of elliptic functions, with the relevant monodromy matrices given by quasi-permutations. Their higher rank gene\-ra\-lization was introduced in \cite{Korotkin} without imposing an extra condition of semi-degenerate monodromy. There is, however, a simple instance of quasi-permutation monodromy given by permutation of only two basis vectors. Its spectrum $\{1,\ldots,1,-1\}$ satisfies the semi-degeneracy condition. An interesting class of algebro-geometric FST$_N$ solutions is thus associated to monodromy corresponding to 2 elementary permutations and 2 arbitrary quasi-permutations.
\item \textit{Algebraic solutions}. This class corresponds to finite orbits of an action of the pure braid group~$\mathcal P_3$ on the PVI monodromy manifold~\cite{DM}. The classification of such orbits was accomplished in~\cite{LT}. In Section~\ref{sec_BG} of the present work, we describe the algebraic braid group dynamics on FST$_N$ monodromy data and set up a classification problem to be solved.
\end{itemize}

The solution of Painlev\'e VI corresponding to monodromy in general position was expressed as an explicit Fredholm determinant in~\cite{GL16}; see~\cite{CGL} for substantially simplified proof. The relevant integral operator acts on $\mathbb C^2\otimes L^2\big(S^1\big)$ and its integrable kernel is written in terms of solutions of two hypergeometric systems. The determinant expansion further yields a series representations for the PVI tau function. The present work extends both determinant and series representations to the case of the FST$_N$ system (Sections~\ref{secdet} and~\ref{secseries}).

When one of the auxiliary systems has reducible monodromy and one of the relevant monodromy matrices generates a nilpotent subgroup of $\mathrm{SL}(2,\Cb)$, the block integral kernel of \cite[Theorem~A]{GL16} can be reduced to a scalar continuous hypergeometric kernel on an interval. The corresponding PVI solution first appeared in \cite{BD}; it generalizes the well-known sine- and Airy-kernel Painlev\'e~V and~II transcendents. We obtain an FST$_N$ analog of this solution in Section~\ref{sec_fred}. It would be interesting to understand whether, similarly to the $N=2$ case \cite{BO}, it plays a role in the harmonic analysis on infinite groups such as $\mathrm{U}(\infty)$.

The last section of this manuscript is devoted to a multivariate generalization of the FST$_N$ system, denoted by $\mathcal G_{N,n-3}$ in \cite{Tsuda}. The corresponding Fuchsian system has $n$ regular singular points, of which all but 2 have semi-degenerate spectral type $(N-1,1)$. We present the Fredholm determinant and series representation of the corresponding tau function and identify the coefficients of the latter with Nekrasov functions \cite{Nekrasov} for $\mathcal N=2$ $\mathrm{U}(N)$ linear quiver gauge theory. In combination with the results of \cite{GIL18}, this gives a direct proof of the AGT-W correspondence for $c=N-1$ (Appendix~\ref{appendixA}).

\section[Semi-degenerate Fuchsian system and FST$_N$]{Semi-degenerate Fuchsian system and FST$\boldsymbol{{}_N}$}
This section explains the relation of FST$_N$ system and monodromy preserving deformations. The relevant results are essentially extracted from \cite{Tsuda}.

Let $G=\mathrm{GL}(N,\Cb)$. Consider a Fuchsian system with 4 regular singular points on $\Cb\Pb^1$,
\begin{gather}\label{lis}
\partial_z\Phi=\Phi A(z),\qquad A(z)=\sum_{\nu=0,1,t}\frac{A_{\nu}}{z-\nu}, \qquad \text{with} \quad A_{0,1,t}\in \operatorname{Mat}_{N\times N}(\Cb).
\end{gather}
Denote $A_{\infty}:=-\sum\limits_{\nu=0,1,t}A_{\nu}$. Any locally defined fundamental matrix solution $\Phi(z)$ can be analytically continued to any simply connected domain in $\mathbb C\backslash \{0,1,t \}$. Its global analytic properties are encoded in the monodromy representation $m\in \operatorname{Hom}(\pi_1(\mathbb C\backslash \{0,1,t \}),G)$ associated to~\eqref{lis}. Choosing the generators $\gamma_{0,1,t,\infty}$ of the fundamental group as shown in Fig.~\ref{figg1}, the monodromy of~$\Phi(z)$ is described by their 4 images $M_{\nu}=m(\gamma_{\nu})$ which satisfy $M_0M_tM_1M_{\infty}=\mathds 1$. The freedom to choose an arbitrary basis of solutions means that the linear system~\eqref{lis} only defines the conjugacy class of monodromy.
\begin{figure}[h!]\centering\vspace{-9mm}
\begin{tikzpicture}[baseline,yshift=-0.5cm,scale=0.45]
\draw (-4,0) .. controls (-1,0.5) and (-0.5,1) .. (0,1)
.. controls +(0.5,0) and +(0,1) .. (1,0);
\draw[->,>=latex] (-4,0) .. controls (-1,-0.5) and (-0.5,-1) .. (0.2,-1);
\draw (0,-1)
.. controls +(0.5,0) and +(0,-1) .. (1,0);
\draw[fill] (0,0) circle (0.06);
\draw[fill] (2.5,0) circle (0.06);
\draw[fill] (5.5,0) circle (0.06);
\draw[->,>=latex] (-4,0) .. controls (0,1.5) and (1.5,2.5) .. (1.5,0) .. controls (1.5,-0.5) and (2,-1) .. (2.7,-1);
\draw (2.5,-1) .. controls (3,-1) and (3.5,-0.5) .. (3.5,0);
\draw (-4,0) .. controls (0,2.5) and (3.5,2.5) .. (3.5,0);
\draw[->,>=latex] (-4,0) .. controls (0,3) and (4.5,3) .. (4.5,0) .. controls (4.5,-0.5) and (5,-1) .. (5.7,-1);
\draw (5.5,-1) .. controls (6,-1) and (6.5,-0.5) .. (6.5,0);
\draw (-4,0) .. controls (0,3.5) and (6.5,4.5) .. (6.5,0);
\draw[->,>=latex] (-4,0) .. controls (0,4.5) and (8,4.5) .. (8,0);
\draw (-4,0) .. controls (0,-4.5) and (8,-4.5) .. (8,0);
\draw (0.2,-1.1) node[below] {$\gamma_0$};
\draw (2.7,-1.1) node[below] {$\gamma_t$};
\draw (5.7,-1.1) node[below] {$\gamma_1$};
\draw (8,0) node[right] {$\gamma_\infty$};
\draw (0.1,0.1) node[below left] {$0$};
\draw (2.6,0.1) node[below left] {$t$};
\draw (5.6,0.1) node[below left] {$1$};
\end{tikzpicture}\vspace{-7mm}
\caption{Generators $ \{\gamma_{\nu} \}$.}\label{figg1}
\end{figure}
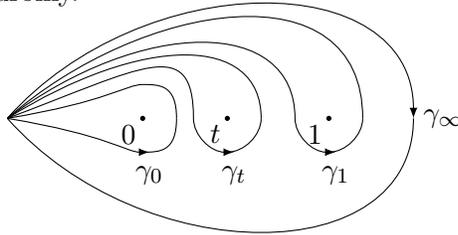

If the eigenvalues of each $A_{\nu}$ ($\nu=0,1,t,\infty$) do not differ by non-zero integers, the monodromy preserving deformation of \eqref{lis} is described by the equations
\begin{gather}
\frac{{\rm d}A_0}{{\rm d}t}=\frac{[A_0,A_t]}{t}+[A_0,K],\qquad \frac{{\rm d}A_1}{{\rm d}t}=\frac{[A_1,A_t]}{t-1}+[A_1,K],\nonumber\\ \frac{{\rm d}A_t}{{\rm d}t}=-\frac{[A_0,A_t]}{t}-\frac{[A_1,A_t]}{t-1}+[A_t,K],\qquad K=L^{-1}\frac{{\rm d}L}{{\rm d}t}.\label{deformeq}
\end{gather}
Here $L(t)\in\mathrm{GL}(N,\Cb)$ is an arbitrary invertible matrix function of $t$ corresponding to the choice of the gauge of $\Phi(z)$ at $z=\infty$. It can be absorbed by setting $A_{\nu}=L^{-1}\tilde A_{\nu} L$. Although one may be tempted to set $L(t)=\operatorname{const}$, $K=0$ from the very beginning, below we will also use a~more subtle time-dependent gauge. The deformation equations~\eqref{deformeq} appear as the compatibility conditions of the Fuchsian system with
\begin{gather*}
\partial_t\Phi=\Phi B(z),\qquad B(z)= K-\frac{A_t}{z-t}.
\end{gather*}

The spectra $\operatorname{\mathrm{Spec}}(A_{\nu})$ give exponents of local monodromy around $z=\nu$ and are therefore conserved under isomonodromic evolution. The main object of interest for us will be the Jimbo--Miwa--Ueno tau function \cite{JMU} defined by{\samepage
\begin{gather}\label{jmutau}
\frac{{\rm d}}{{\rm d}t}\ln\tau_{\mathrm{JMU}}(t):=\frac12\operatorname{Res}_{z=t}\operatorname{Tr} A^2(z)=\frac{\operatorname{Tr}A_0A_t}{t}+\frac{\operatorname{Tr}A_1A_t}{t-1}.
\end{gather}
The tau function is clearly $L$-independent.}

We are now going to parameterize the matrices $A_{\nu}$ and rewrite the deformation equations~\eqref{deformeq} more explicitly in the semi-degenerate case. Let the singular points $z=1,t$ be of spectral type $(N-1,1)$. Employing if necessary a scalar gauge transformation of the form $\Phi(z) \mapsto (z-t)^a(z-1)^b \tilde{\Phi}(z)$, it can be assumed that the eigenvalue of $A_t$ and $A_1$ with multiplicity $N-1$ is equal to~$0$; the non-degenerate eigenvalues will be denoted by~$\Lambda_t$ and~$\Lambda_1$. Assume in addition that $A_0$ and $A_{\infty}$ remain sufficiently generic so that they can be brought to a lower and upper triangular form:
\begin{gather*}
A_0=\left(\begin{matrix}
\theta_0^{(1)} & 0 & \cdots & 0 \\
* & \theta_0^{(2)} & \cdots & \cdot \\
\vdots & \vdots & \ddots & 0 \\
* & \cdot & * & \theta_0^{(N)}
\end{matrix}\right),\qquad
A_{\infty}=\left(\begin{matrix}
\theta_{\infty}^{(1)} & * & \cdots & * \\
0 & \theta_{\infty}^{(2)} & \cdots & \cdot \\
\vdots & \vdots & \ddots & * \\
0 & \cdot & 0 & \theta_{\infty}^{(N)}
\end{matrix}\right).
\end{gather*}
Write $A_{t,1}$ in this basis as $A_t=-q^T\otimes p$, $A_1=-\bar q^T\otimes\bar p$. Here $q$, $\bar q$, $p$, $\bar p$ are row vectors. The notation for their individual entries is fixed so that, e.g., $q=\big(
q^{(1)} \,\ldots \, q^{(N)}\big)$. We have the constraints
\begin{gather}\label{bcconstraints}
\Lambda_t=-\sum_{k=1}^N p^{(k)}q^{(k)},\qquad \Lambda_1=-\sum_{k=1}^N \bar p^{(k)}\bar q^{(k)},\qquad
\Lambda_t+\Lambda_1+\sum_{k=1}^N\big(\theta_0^{(k)}+ \theta_{\infty}^{(k)}\big)=0.
\end{gather}
The former two conditions come from the computation of $\operatorname{Tr}A_{t,1}$ and the 3rd is the Fuchs relation. The triangular form of $A_{0,\infty}$ and the relation $\sum\limits_{\nu=0,1,t,\infty}A_{\nu}=0$ imply that
\begin{gather*}
(A_0)_{lk} = p^{(k)}q^{(l)}+\bar p^{(k)} \bar q^{(l)},\qquad k<l,\\
(A_{\infty})_{lk}= p^{(k)}q^{(l)}+\bar p^{(k)} \bar q^{(l)},\qquad k>l.
\end{gather*}
If all $\bar q^{(k)}\ne0$, the remaining freedom of diagonal gauge transformations can be used to set $\bar q^{(k)}=1$. Considering the diagonal of the same matrix relation $\sum\limits_{\nu=0,1,t,\infty}A_{\nu}=0$, we then obtain
\begin{gather*}
\bar p^{(k)}=-p^{(k)}q^{(k)}+\theta_{\infty}^{(k)}+\theta_0^{(k)},\qquad k=1,\ldots, N.
\end{gather*}
We can further set $q^{(1)}=1$ and express $p^{(1)}$ using the 1st of relations~\eqref{bcconstraints}, so that $A(z)$ is completely parameterized by $2(N-1)$ quantities $p^{(k)}$, $q^{(k)}$ with $k=2,\ldots,N$.

To rewrite the equations \eqref{deformeq} in these coordinates, we also need the expression for $K$. Note that since $A_{\infty}$ is upper triangular, so are $H$ and $K$. On the other hand, the mat\-rix $B(z=0)=K+\frac{A_t}{t}=G_0^{-1}\partial_tG_0$, where $G_0$ characterizes the gauge at $z=0$: namely, $A_0=G_0^{-1}\operatorname{diag}\bigl\{\theta_0^{(1)},\ldots,\theta_0^{(N)}\bigr\}G_0$. Therefore, $B(z=0)$ is lower triangular, which immediately yields the off-diagonal elements of~$K$:
\begin{gather*}
K_{kl}=\begin{cases}
\displaystyle \frac{q^{(k)} p^{(l)}}{t}, &k<l,\\
 0, & k>l.
\end{cases}
\end{gather*}
In order to determine the diagonal, one can similarly use that $B(z=1)=K+\frac{A_t}{t-1}=G_1^{-1}\partial_t G_1$, where $A_1=G_1^{-1}\operatorname{diag} \{\Lambda_1,0,\ldots,0 \}
G_1$. In the gauge where all $\bar q^{(k)}=1$, we can set $\big(G_1^{-1}\big)_{k1}=1$ for $k=1,\ldots, N$, so that $\sum\limits_{k=1}^N\big(K+\frac{A_t}{t-1}\big)_{lk}=0$. It follows that
\begin{gather*}
K_{ll}=-\frac{q^{(l)}}{t}\sum_{k=l+1}^N p^{(k)}+\frac{q^{(l)}}{t-1} \sum_{k=1}^N p^{(k)},\qquad l=1,\ldots,N.
\end{gather*}

The deformation equations \eqref{deformeq} can now be rewritten as follows. For $k=2,\ldots, N-2$, we have
\begin{gather}
\frac{{\rm d}q^{(k)} }{{\rm d}t}= \sum_{l=1}^N(f\delta_{kl}-g_{kl}) q^{(l)} ,\nonumber\\
\frac{{\rm d}p^{(k)} }{{\rm d}t}= \sum_{l=1}^N p^{(l)}(g_{lk}-f\delta_{kl}) ,\label{fss}
\end{gather}
where $g=\frac{A_0}{t}+\frac{A_1}{t-1}+K$ and
\begin{gather*}
f=\sum_{m=1}^N g_{1m}q^{(m)}=\frac{\theta_0^{(1)}-\Lambda_t}{t}+\frac{\sum\limits_{m=1}^N p^{(m)}}{t(t-1)}-\frac{\sum\limits_{m=1}^N\bar p^{(m)}q^{(m)}}{t-1}.
\end{gather*}
For $k=1$, the equation \eqref{fss} is satisfied automatically (recall that $q^{(1)}=1$ and $p^{(1)}=-\Lambda_t-\sum\limits_{k=2}^Np^{(k)}q^{(k)}$).
More explicitly, the isomonodromic evolution \eqref{fss} of the local coordina\-tes~$p^{(k)}$,~$q^{(k)}$ is described by
\begin{gather*}
\frac{{\rm d}q^{(k)} }{{\rm d}t}= \left( f +\frac{\Lambda_t -\theta_0^{(k)}}{t}\right) q^{(k)}+\frac{\sum\limits_{m=1}^N\bar p^{(m)}q^{(m)}}{t-1}\\
\hphantom{\frac{{\rm d}q^{(k)} }{{\rm d}t}=}{} -\frac{\sum\limits_{m<k}\bar p^{(m)}q^{(m)}}{t}-\left[\frac{\sum\limits_{m< k}p^{(m)}}t+\frac{\sum\limits_{m=1}^N p^{(m)}}{t(t-1)}\right]\big(q^{(k)}\big)^2, \\
\frac{{\rm d}p^{(k)} }{{\rm d}t}=\left(\frac{\theta_0^{(k)}-\Lambda_t }{t}-f \right) p^{(k)}-\frac{\sum\limits_{m=1}^Np^{(m)}}{t-1} \bar p^{(k)}\\
\hphantom{\frac{{\rm d}p^{(k)} }{{\rm d}t}=}{} +\frac{\sum\limits_{m>k}p^{(m)}}{t} \bar p^{(k)}
+\left[\frac{\sum\limits_{m< k}p^{(m)}}t+\frac{\sum\limits_{m=1}^N p^{(m)}}{t(t-1)}\right]p^{(k)}q^{(k)}.
\end{gather*}
The latter system can also be rewritten in the (non-autonomous) Hamiltonian form
\begin{gather*}
s\frac{{\rm d}q^{(k)} }{{\rm d}s}= \frac{\partial H}{\partial p^{(k)} } ,\qquad s\frac{{\rm d}p^{(k)} }{{\rm d}s}=- \frac{\partial H}{\partial q^{(k)} } , \qquad k=2,\ldots,N,
\end{gather*}
where the Hamiltonian function is given by
\begin{gather*}
H=\sum_{m=1}^N \theta_0^{(m)}p^{(m)}q^{(m)}+\sum_{1\le m<n\le N} \big(\bar p^{(m)}+p^{(m)}q^{(n)}\big) p^{(n)}q^{(m)}+\frac{1}{s-1}\sum_{m,n=1}^N\bar p^{(m)}p^{(n)}q^{(m)}.
\end{gather*}
This is the Fuji--Suzuki--Tsuda system FST$_N$. Its tau function $\tau_{\mathrm{FST}}(t)$, obtained by the appropriate specialization of the Jimbo--Miwa--Ueno tau differential~\eqref{jmutau}, is directly related to the Hamiltonian by $s\frac{{\rm d}}{{\rm d}s}\ln\tau_{\mathrm{FST}}\big(s^{-1}\big)=H$.

In the following sections, we focus our attention on the computation of $\tau_{\mathrm{FST}}(t)$. This function is fundamental from the point of view of applications. The local coordinates on the FST$_N$ phase space, such as $p^{(k)}$, $q^{(k)}$, can be expressed in terms of tau functions~\cite{Tsuda0,Tsuda} similarly to the $N=2$ case which corresponds to Painlev\'e VI.

\section{Braid group dynamics of semi-degenerate monodromy\label{sec_BG}}
The Riemann--Hilbert correspondence assigns to each linear system \eqref{lis} its monodromy~-- a~point in the space~$\mathcal M$ of conjugacy classes of triples of $G$-matrices:
\begin{gather*}
\mathcal M=\operatorname{Hom}(\pi_1(\mathbb C\backslash \{0,1,t \}),G)/G\cong \big\{ [( M_0,M_t,M_1)], M_{0,t,1}\in G\big\}.
\end{gather*}
Local solutions of the deformation equations~\eqref{deformeq} continue to analytic functions of $t$ on the universal covering of $\mathbb C\backslash \{0,1 \}$. This analytic continuation may be described~\cite{DM} in terms of an action of~$\mathcal B_3$, the braid group on 3~strands, on~$\mathcal M$. Its origin may be explained as follows: the center $\mathcal Z\subset \mathcal B_3$ in fact acts trivially and
$\mathcal B_3/\mathcal Z\cong \Gamma=\mathrm{PSL}(2,\mathbb Z)$ is the mapping class group of $\mathbb C\backslash \{0,1,t \}$, related to the outer automorphisms of the corresponding $\pi_1$ by Dehn--Nielsen theorem.

To describe the braid/modular group action explicitly in a compact way, it is convenient \cite{LT} to first extend it to an action of $\bar\Gamma=\mathrm{PGL}(2,\mathbb Z)$ and restrict the latter to the level 2 congruence subgroup $\bar\Lambda=\bar\Gamma(2)$. The principal gain is an extremely simple presentation of $\bar\Lambda$, which is freely generated by 3 involutions:
\begin{gather*}
\bar{\Lambda}=\big \langle x,y,z\,|\,x^2=y^2=z^2=1\big\rangle.
\end{gather*}
The action of these generators on the representative triples is as follows:
\begin{gather}
x\colon \ \big[(M_0,M_t,M_1)\big]\mapsto\big[\big( M_0^{-1},M_t^{-1},M_t M_1^{-1}M_t^{-1}\big)\big],\nonumber\\
y\colon \ \big[(M_0,M_t,M_1)\big]\mapsto\big[\big( M_1 M_0^{-1}M_1^{-1},M_t^{-1},M_1^{-1}\big)\big],\nonumber\\
z\colon \ \big[(M_0,M_t,M_1)\big]\mapsto\big[\big( M_0^{-1}, M_0 M_t^{-1}M_0^{-1},M_1^{-1}\big)\big] .\label{xyzd}
\end{gather}
The words of even length in $x$, $y$, $z$ form the subgroup $\Lambda\subset\bar\Lambda$. It is isomorphic to free group on 2~generators (e.g., $xy$ and~$yz$) and can be seen as a quotient $\mathcal P_3/\mathcal Z\cong \Lambda$ of pure braids.

Solutions of the deformation equations~\eqref{deformeq} with finite branching (in particular, algebraic ones) correspond to finite $\bar\Lambda$-orbits on $\mathcal M$. Of course, many examples of such orbits can be easily produced by monodromies generating finite subgroups of $\mathrm{GL}(N,\Cb)$. However, their classification is an open problem except for $N=2$~\cite{LT}. In the latter case, assuming that $M_{\nu}\in\mathrm{SL}(2,\Cb)$ for $\nu=0,t,1,\infty$, the local conjugacy classes $ [M_{\nu} ]$ are preserved by the action of $\bar\Lambda$. The classification problem then reduces to the description of periodic orbits generated by 3 polynomial involutive automorphisms of an affine cubic in $\Cb^3$ representing the $\mathrm{SL}(2,\Cb)$-character variety of the 4-punctured sphere.

We are now going to consider semi-degenerate monodromy and describe its dynamics in a~more explicit form. The procedure is somewhat reminiscent of what was done on the other side of the Riemann--Hilbert correspondence in the previous subsection. Assume that $M_0$ and $M_{\infty}$ are diagonalizable, and $M_{1}-\mathds 1$ and $M_{t}-\mathds 1$ have rank~1. Write
\begin{gather}\label{Mt1}
M_t=\mathds 1+u^T\otimes v,\qquad M_1=\mathds 1+\bar u^T\otimes \bar v ,
\end{gather}
and also denote
\begin{gather*}
 [M_{\nu} ]= [\exp(2\pi{\rm i} \Theta_{\nu}) ],\qquad \nu=0,1,t,\infty,\qquad
\Theta_{\nu}=\begin{cases}\operatorname{diag}\big(\theta_{\nu}^{(1)},\ldots, \theta_{\nu}^{(N)}\big), &\nu=0,\infty,\\
\operatorname{diag}(\Lambda_{\nu},0,\ldots,0), & \nu=t,1.
\end{cases}
\end{gather*}
In this notation, we have $v\cdot u^T={\rm e}^{2\pi {\rm i} \Lambda_t}-1$ and $\bar v\cdot \bar u^T={\rm e}^{2\pi{\rm i} \Lambda_1}-1$, which implies that $M_t^{-1}=\mathds 1-{\rm e}^{-2\pi{\rm i}\Lambda_t}u^T\otimes v$ and $M_1^{-1}=\mathds 1-{\rm e}^{-2\pi{\rm i}\Lambda_1}\bar u^T\otimes \bar v$. Observe that the only effect of the $\bar \Lambda$-action on the local monodromy exponents is the sign flip of all $\Theta_{\nu}$ for the words of odd length in $x$, $y$, $z$.

Let us make a further assumption that $M_0M_t=( M_1 M_{\infty})^{-1} $ has distinct eigenvalues, different from those of $M_0$ and $M_{\infty}^{-1}$. Pick a representative in the conjugacy class of monodromy such that
\begin{gather}
M_0M_t=\exp(2\pi{\rm i}\mathfrak S), \qquad \mathfrak S=\operatorname{diag}\big\{\sigma^{(1)},\ldots,\sigma^{(N)}\big\},\nonumber\\
\operatorname{Tr}\mathfrak S=\operatorname{Tr}\Theta_0+\Lambda_t=-\operatorname{Tr}\Theta_{\infty}-\Lambda_1.\label{M0Mt}
\end{gather}
\begin{Lemma}\label{lem31} For $k=1,\ldots,N$, we have
\begin{subequations}\label{uvrels1}
\begin{gather}\label{uvrels1a}
u^{(k)}v^{(k)}= {\rm e}^{2\pi{\rm i}(\Lambda_t-\sigma^{(k)})}\frac{\prod\limits_{l=1}^N\big({\rm e}^{2\pi{\rm i} \sigma^{(k)}}-{\rm e}^{2\pi{\rm i} \theta_0^{(l)}}\big)}{\prod\limits_{l\ne k}\big({\rm e}^{2\pi{\rm i} \sigma^{(k)}}-{\rm e}^{2\pi{\rm i} \sigma^{(l)}}\big)},\\
\label{uvrels1b}\bar u^{(k)}\bar v^{(k)}= -{\rm e}^{-2\pi{\rm i} \sigma^{(k)}}\frac{\prod\limits_{l=1}^N\big({\rm e}^{2\pi{\rm i} \sigma^{(k)}}-{\rm e}^{-2\pi{\rm i} \theta^{(l)}_{\infty}}\big)}{\prod\limits_{l\ne k}\big({\rm e}^{2\pi{\rm i} \sigma^{(k)}}-{\rm e}^{2\pi{\rm i} \sigma^{(l)}}\big)}.
\end{gather}
\end{subequations}
\end{Lemma}
\begin{proof}Note that
\begin{gather*}
\det ( \lambda-M_0)=\det \big(\lambda-M_0 M_t+{\rm e}^{-2\pi{\rm i}\Lambda_t} M_0 M_t u^{T}\otimes v\big)\nonumber\\
\hphantom{\det ( \lambda-M_0)}{}=
\det (\lambda-M_0 M_t)\left[1+{\rm e}^{-2\pi{\rm i} \Lambda_t}v\frac{M_0M_t}{\lambda - M_0M_t}u^T\right].
\end{gather*}
Taking the limit as $\lambda\to {\rm e}^{2\pi{\rm i}\sigma^{(k)}}$, we obtain the first of relations~\eqref{uvrels1}. The second is obtained by similar considerations from $\det \big( \lambda-M_{\infty}^{-1}\big)=\det (\lambda-M_0M_tM_1)$.
\end{proof}

The remaining freedom of diagonal conjugation allows for rescalings
\begin{gather*}
\big(u^{(k)},v^{(k)},\bar u^{(k)},\bar v^{(k)}\big)\mapsto \big(\alpha_k u^{(k)},\alpha_k^{-1}v^{(k)},\alpha_k\bar u^{(k)},\alpha_k^{-1}\bar v^{(k)}\big),
\end{gather*}
with arbitrary $\alpha_k\in\Cb^*$. It can be used, for example, to fix all $\bar u^{(k)}$ to be $1$ and express all $\bar v^{(k)}$ from~\eqref{uvrels1b}. The only ambiguity which remains afterwards is the overall rescaling $(u,v)\mapsto \big(\alpha u,\alpha^{-1} v\big)$ with $\alpha\in\Cb^*$. This suggests to introduce the notation
\begin{gather}\label{skldef}
\sigma_{k,l}:=\sigma^{(k)}-\sigma^{(l)},\qquad s_{k,l}:=\frac{u^{(k)}\bar u^{(l)}}{\bar u^{(k)}u^{(l)}},\qquad k,l=1,\ldots,N.
\end{gather}
The local coordinates on the moduli space of semi-degenerate monodromy can be chosen as any $N-1$ independent $\sigma_{k,l}$ and $N-1$ independent $s_{k,l}$; one admissible option is to take $l=k+1$ with $k=1,\ldots,N-1$. These coordinates unambiguously fix the conjugacy class of $(M_t,M_1,M_0M_t)$.
\begin{Lemma}The involutions $x$, $z$ from \eqref{xyzd} act as
\begin{gather}\label{xzaction}
x\colon \ (\mathfrak S,s_{kl})\mapsto (-\mathfrak S,s_{kl}),\qquad z\colon \ (\mathfrak S,s_{kl}) \mapsto \bigl({-}\mathfrak S, {\rm e}^{2\pi{\rm i} \sigma_{kl}}s_{kl}\bigr).
\end{gather}
\end{Lemma}
\begin{proof} We have
\begin{gather*} x( [(M_0,M_t,M_1) ])=\big[\big(M_t^{-1}M_0^{-1}M_t,M_t^{-1},M_1^{-1}\big)\big].\end{gather*} Since the product $M_t^{-1}M_0^{-1}M_t M_t^{-1}={\rm e}^{-2\pi{\rm i}\mathfrak S}$ is a diagonal matrix, it follows that $x(\mathfrak S)=-\mathfrak S$. On the other hand, the transformation of the vectors $u$, $v$, $\bar u$, $\bar v$ can then be written as
\begin{gather*}
x\colon \ ( u,v,\bar u,\bar v) \mapsto \big(u,-{\rm e}^{-2\pi{\rm i} \Lambda_t}v,\bar u,-{\rm e}^{-2\pi{\rm i} \Lambda_1}\bar v\big).
\end{gather*}
Applying the same reasoning to $z( [( M_0,M_t,M_1) ])=\big[\big( M_0^{-1}, M_0 M_t^{-1}M_0^{-1},M_1^{-1}\big)\big] $, we see that
\begin{gather*}
z\colon \ ( u,v,\bar u,\bar v )\mapsto \big(u {\rm e}^{2\pi{\rm i} \mathfrak S},-{\rm e}^{-2\pi{\rm i} \Lambda_t}v {\rm e}^{-2\pi{\rm i}\mathfrak S},\bar u, -{\rm e}^{-2\pi{\rm i} \Lambda_1}\bar v\big),
\end{gather*}
which yields the second equation in \eqref{xzaction}.
\end{proof}
\begin{Remark}
If the $\bar \Lambda$-orbit is finite, there exists an $n\in\mathbb N$ such that the action of $(xz)^n$ leaves invariant the conjugacy class of monodromy. Under the above genericity assumptions, \eqref{xzaction}~then implies that $\sigma_{kl}\in \mathbb Q$. This provides a systematic way to look for examples of finite $\bar \Lambda$-orbits not bounded to those coming from the finite monodromy groups. Specifically, one can set $ [M_0M_t ]=\big[{\rm e}^{2\pi{\rm i}\mathfrak S}\big]$, $ [M_0M_1 ]=\big[{\rm e}^{2\pi{\rm i}\bar{\mathfrak S}}\big]$ and perform an exhaustive computer search of orbits whose all points are characterized by $\sigma_{kl},\bar{\sigma}_{kl}\in\mathbb Q$ with sufficiently small denominators. In the $N=2$ case, this procedure very quickly gives the list of all exceptional finite orbits; the proof of its completeness is quite tedious, though.
\end{Remark}

It will become clear in the next subsections that the rather simple form of the $x$- and $z$-transformation reflects the fact that $xz$ corresponds to analytic continuation of solution of the FST$_N$ system around the branch point $t=0$, and the coordinates $( \{\sigma_{kl} \}, \{s_{kl} \})$ on $\mathcal M$ are well-adapted for the description of the corresponding local behavior. The form of $y$-transformation is more involved. In order to describe it more explicitly, let us first formulate an auxiliary lemma.
\begin{Lemma}\label{lemrig}Let $a=\big\{a^{(k)}\big\}_{k=1,\ldots, N}$ and $b=\big\{b^{(k)}\big\}_{k=1,\ldots, N}$ such that $a\cap b=\varnothing$ and $a^{(i)}\ne a^{(j)}$, $b^{(i)}\ne b^{(j)}$ for $i\ne j$,
\begin{gather*}
M_{a}=\operatorname{diag}\big\{a^{(1)},\ldots,a^{(N)}\big\},\qquad M_{b}=\operatorname{diag}\big\{b^{(1)},\ldots,b^{(N)}\big\},
\end{gather*}
such that and $a^{(i)}\ne a^{(j)}$,
$b^{(i)}\ne b^{(j)}$ for $i\ne j$. Define $ W_{M_a\to M_b}\in\operatorname{Mat}_{N\times N}(\Cb)$ by
\begin{gather}\label{Wdef}
\big( W_{M_a\to M_b}\big)_{ij}=\prod_{k\ne i}\frac{a^{(j)}-b^{(k)}}{b^{(i)}-b^{(k)}}, \qquad i,j=1,\ldots,N.
\end{gather}
We have
\begin{gather}
M_{b}^{-1}W_{M_a\to M_b}M_{a} \bigl( W_{M_a\to M_b}\bigr)^{-1}=\mathds 1+u_{M_a,M_b}^T\otimes v_{M_a,M_b},\nonumber\\
u_{M_a,M_b}^{(i)}=\bigg[b^{(i)}\prod_{k\ne i}\big(b^{(i)}-b^{(k)}\big)\bigg]^{-1},\qquad v_{M_a,M_b}^{(i)}= -\prod_{k}\big(b^{(i)}-a^{(k)}\big).\label{Wab}
\end{gather}
Moreover, $\big( W_{M_a\to M_b}\big)^{-1}=W_{M_b\to M_a}$.
\end{Lemma}
\begin{proof} The statements can be verified directly using Lagrange interpolation. \end{proof}

\begin{Remark}We have seen previously that, if two square matrices $M$, $M'$ have disjoint simple spectra and $M^{-1}M'-\mathds 1$ is of rank~1, then the conjugacy class $ [( M,M')]$ is uniquely determined by the eigenvalues of $M$ and $M'$. Lemma~\ref{lemrig} determines, up to diagonal factors, the explicit form of the matrix~$W$ relating the eigenbases of~$M$ and~$M'$. In particular,the relations~\eqref{uvrels1a} and~\eqref{uvrels1b} in Lemma~\ref{lem31} can be obtained as respective corollaries of \eqref{Wab} with $(M_{a},M_{b})=\big({\rm e}^{2\pi{\rm i}\Theta_0},{\rm e}^{2\pi{\rm i}\mathfrak S}\big)$ and $(M_{a},M_{b}) =\big({\rm e}^{-2\pi{\rm i}\Theta_{\infty}},{\rm e}^{2\pi{\rm i}\mathfrak S}\big)$.
\end{Remark}

More importantly, in the previous basis where $M_0M_t$ is diagonal, the matrix $M_0^{-1}$ can be represented as
\begin{gather*}
M_0^{-1}=D_{{\rm e}^{2\pi{\rm i}\Theta_0}\to {\rm e}^{2\pi{\rm i} \mathfrak S}}W_{{\rm e}^{2\pi{\rm i}\Theta_0}\to {\rm e}^{2\pi{\rm i} \mathfrak S}}{\rm e}^{-2\pi{\rm i} \Theta_0}W_{{\rm e}^{2\pi{\rm i}\Theta_0}\to {\rm e}^{2\pi{\rm i} \mathfrak S}}^{-1}D_{{\rm e}^{2\pi{\rm i}\Theta_0}\to {\rm e}^{2\pi{\rm i} \mathfrak S}}^{-1},
\end{gather*}
where $W_{{\rm e}^{2\pi{\rm i}\Theta_0}\to {\rm e}^{2\pi{\rm i} \mathfrak S}}$ is defined by \eqref{Wdef} and $D_{{\rm e}^{2\pi{\rm i}\Theta_0}\to {\rm e}^{2\pi{\rm i} \mathfrak S}}$ is a diagonal matrix whose non-zero elements can be found from the relation $M_t=\mathds 1+u^T\otimes v$ up to irrelevant overall scaling. One may fix them, e.g., by choosing
\begin{gather*}
D_{{\rm e}^{2\pi{\rm i}\Theta_0}\to {\rm e}^{2\pi{\rm i} \mathfrak S}}u_{{\rm e}^{2\pi{\rm i}\Theta_0},{\rm e}^{2\pi{\rm i} \mathfrak S}}^T= u^T.
\end{gather*}
Let us now consider the matrix
\begin{gather*}
\mathfrak P:=y\big({\rm e}^{2\pi{\rm i} \mathfrak S}\big),\qquad \mathfrak P=\operatorname{diag}\big\{\rho^{(1)},\ldots,\rho^{(N)}\big\},
\end{gather*}
defined up to permutation of its diagonal entries.

\begin{Lemma}Suppose that $\operatorname{Spec}\big({\rm e}^{-2\pi{\rm i} \Theta_0}\big)\cap \operatorname{Spec}\big( M_1 M_0^{-1}M_1^{-1}M_t^{-1}\big)=\varnothing $. The coefficients of the characteristic polynomial $\det (\rho-\mathfrak P)$ can be explicitly written as rational functions in $\big\{{\rm e}^{2\pi{\rm i} \sigma_{kl}}\big\}$, $ \{s_{kl} \}$.
\end{Lemma}
\begin{proof} The spectrum of $\mathfrak P$ is determined by the equation $\det \big(\rho-M_1 M_0^{-1}M_1^{-1}M_t^{-1}\big)=0 $. Under genericity assumptions of the lemma, this can be rewritten as $\det (\mathds 1+L)=0$, where
\begin{gather*}
L=\big(\rho {\rm e}^{2\pi{\rm i} \Theta_0}-1\big)^{-1}\big[ \mathds 1-\tilde W^{-1}M_t^{-1}\tilde W\big],\qquad \tilde W=M_1 D_{{\rm e}^{2\pi{\rm i}\Theta_0}\to {\rm e}^{2\pi{\rm i} \mathfrak S}}W_{{\rm e}^{2\pi{\rm i}\Theta_0}\to {\rm e}^{2\pi{\rm i} \mathfrak S}}.
\end{gather*}
The matrix $L$ has rank 1, which means that $\det (\mathds 1+L)=1+\operatorname{Tr} L$. The equation on the spectrum thus becomes
\begin{gather*}
v\tilde W \big(\rho {\rm e}^{2\pi{\rm i} \Theta_0}-1\big)^{-1} \tilde W^{-1} u^T =-{\rm e}^{2\pi{\rm i} \Lambda_t}.
\end{gather*}
The left hand side of this relation contains
\begin{gather*} vM_1 D_{{\rm e}^{2\pi{\rm i}\Theta_0}\to {\rm e}^{2\pi{\rm i} \mathfrak S}}=v D_{{\rm e}^{2\pi{\rm i}\Theta_0}\to {\rm e}^{2\pi{\rm i} \mathfrak S}} +
\bar v D_{{\rm e}^{2\pi{\rm i}\Theta_0}\to {\rm e}^{2\pi{\rm i} \mathfrak S}} \big(v\bar u^T\big) .\end{gather*}
 The components of the row vector $v D_{{\rm e}^{2\pi{\rm i}\Theta_0}\to {\rm e}^{2\pi{\rm i} \mathfrak S}}$ involve only the products $u^{(k)}v^{(k)}$, which are given by~\eqref{uvrels1a}. The components of the second summand are linear combinations of
$\bar u^{(j)} v^{(j)}u^{(k)}\bar v^{(k)}$, which can be expressed in terms of $\{s_{kl}\}$ by \eqref{skldef} and \eqref{uvrels1}. Similar reasoning can be repeated for
the column vector $ D_{{\rm e}^{2\pi{\rm i}\Theta_0}\to {\rm e}^{2\pi{\rm i} \mathfrak S}} ^{-1}M_1^{-1}u^T$. \end{proof}

At last, let us explain how to compute the effect of $y$-transformation on $\{s_{kl}\}$. Denoting $\tilde M:= M_1 M_0^{-1}M_1^{-1}M_t^{-1}$, one can write ${\rm e}^{2\pi{\rm i}\Theta_0} \tilde W^{-1} \tilde{M}
\tilde W=\tilde W^{-1} M_t^{-1}\tilde W$. Applying again Lemma~\ref{lemrig} with $(M_{\alpha},M_{\beta})=\big( {\rm e}^{-2\pi{\rm i} \Theta_0},\mathfrak P\big)$, we obtain
\begin{gather}\label{TildeM}
\tilde M=\tilde W\tilde D W_{\mathfrak P\to {\rm e}^{-2\pi{\rm i} \Theta_0}}\mathfrak P \bigl(\tilde W\tilde D W_{\mathfrak P\to {\rm e}^{-2\pi{\rm i} \Theta_0}}\bigr)^{-1},
\end{gather}
where $\tilde D$ is the diagonal matrix determined by
$ \tilde D u^T_{\mathfrak P, {\rm e}^{-2\pi{\rm i} \Theta_0}}=\tilde{W}^{-1}u^T$. Now from~\eqref{xyzd} and~\eqref{TildeM} it follows that we can set
\begin{gather*}
y\big(u^T\big)= W_{ {\rm e}^{-2\pi{\rm i} \Theta_0}\to \mathfrak P}u^T,\qquad y\big(\bar u^T\big)=\bigl(\tilde W\tilde D W_{\mathfrak P\to {\rm e}^{-2\pi{\rm i} \Theta_0}}\bigr)^{-1}\bar u^T,
\end{gather*}
which is sufficient to compute $ \{s_{kl} \}$ from~\eqref{skldef}. They can therefore be expressed as rational functions of $\big\{{\rm e}^{2\pi{\rm i} \sigma_{kl}}\big\}$, $ \{s_{kl} \}$ and $\big\{\rho^{(k)}\bigl/\rho^{(l)}\big\}$. Recall, however, that $\mathfrak P$, $\mathfrak S$ and $ \{s_{kl} \}$ are not independent.

The above results can be used for experimental search for the finite $\bar\Lambda$-orbits on $\mathcal M$ with the help of computer algebra. However, a somewhat complicated form of the $y$-transformation is a signal that the local coordinates such as $\{\sigma_{kl}\}$, $\{s_{kl}\}$ are not adapted to the problem of complete classification of such orbits. In addition, various special cases ruled out by our genericity assumptions should be treated separately. The most efficient approach to classification would be to find a ``linearization'' of the mappings~\eqref{xyzd}, i.e., to interpret them as triples of reflections in some auxiliary (possibly infinite-dimensional?) linear space. A similar idea was successfully implemented in \cite{DM} for PVI with special local monodromy.

\section[Fredholm determinant representation of FST$_N$ tau function]{Fredholm determinant representation of FST$\boldsymbol{_N}$ tau function}\label{secdet}
Let us now turn to the evaluation of the FST$_N$ tau function in terms of monodromy data. The representations of $\tau_{\mathrm{FST}}(t)$ we are dealing with in the next sections are of two types: Fredholm determinants and series over $N$-tuples of partitions. In the algebraic case, they can be expected to facilitate the reconstruction of the explicit algebraic solution curves.

Our starting point is a topological decomposition of the isomonodromic tau function, valid for any (not necessarily semi-degenerate) 4-point Fuchsian system, see \cite[equation~(2.35b)]{CGL}:
\begin{gather}
\tau_{\mathrm{JMU}}\left(
\vcenter{\hbox{
\includegraphics[height=7.2ex]{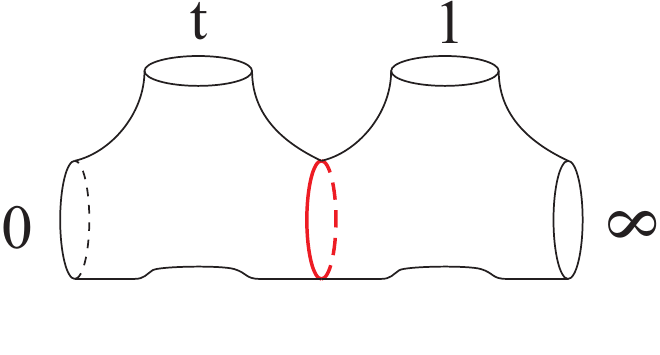}}}\right)=
\tau_{\mathrm{JMU}}\left(
\vcenter{\hbox{
\includegraphics[height=7.2ex]{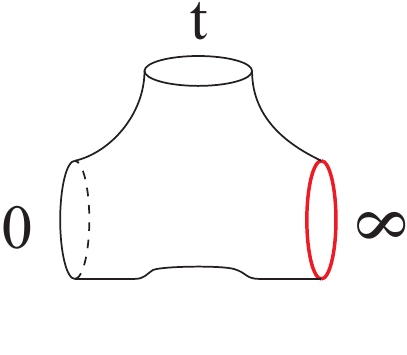}}}\right)
\tau_{\mathrm{JMU}}\left(
\vcenter{\hbox{
\includegraphics[height=7.2ex]{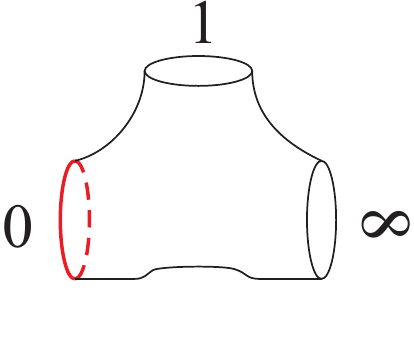}}}\right) \nonumber \\
\hphantom{\tau_{\mathrm{JMU}}\left(
\vcenter{\hbox{
\includegraphics[height=7.2ex]{FSTFigPants}}}\right)=}{} \times
\det \left(\begin{matrix}
\mathds 1 & \mathsf a_{+-}\left(\vcenter{\hbox{
\includegraphics[height=6.4ex]{FSTFigPantsR}}}\right)\\
\mathsf a_{-+}\left(\vcenter{\hbox{
\includegraphics[height=6.4ex]{FSTFigPantsL}}}\right)
& \mathds 1
\end{matrix}\right).\label{topodec}
\end{gather}
The quantities which appear on the right are associated to two auxiliary $3$-point Fuchsian systems with regular singularities at $0$, $t$, $\infty$ and $0$, $1$, $\infty$, whose solutions will be denoted by~$\Phi_-(z)$ and~$\Phi_+(z)$. The monodromy of $\Phi_{\pm}$ is determined by matching the monodromy of the original 4-point system. If we assume for simplicity that $|t|<1$, then the matrix function $\Phi_+^{-1}\Phi$ (resp.~$\Phi_-^{-1}\Phi$) is holomorphic and invertible inside the disk $|z|>t$ (resp. $|z|<1$) on~$\mathbb P^1$. In particular, the monodromy of $\Phi_-$ around $\infty$ and of $\Phi_+$ around $0$ is nothing but the composite monodromy $M_0M_t={\rm e}^{2\pi{\rm i}\mathfrak S}$ of $\Phi$.

The operators $\mathsf a_{\pm\mp}$ are constructed in terms of $\Phi_{\pm}(z)$. Consider the circle $\mathcal C= \{ z\in\Cb\colon $ $ |z|=R,R\in (|t|,1) \}$ and let $H=L^2\big(\mathcal C,\Cb^N\big)$ be the space of vector-valued functions on~$\mathcal C$ seen as Laurent series. Decompose this space as $H=H_+\oplus H_-$, where the subscripts $\pm$ correspond to functions with only positive/negative modes. The operators $\mathsf a_{\pm\mp}\colon H\to H_{\pm}$ are given by
\begin{gather}\label{aops}
(\mathsf a_{\pm\mp}f) (z)=\frac1{2\pi{\rm i}}\oint_{\mathcal C}\mathsf a_{\pm\mp}(z,z') f(z') {\rm d}z',\qquad\mathsf a_{\pm\mp}(z,z')=\pm\frac{\mathds 1-\Psi_{\pm}(z) \Psi_{\pm}^{-1}(z')}{z-z'},
\end{gather}
where $\Psi_{\pm}(z)=(-z)^{-\mathfrak S}\Phi_{\pm}(z)$. Notice that $H_{\pm}\subseteq \operatorname{ker}\mathsf a_{\pm\mp}$, so that $\mathsf a_{\pm\mp}$ can be considered as acting on $H_{\mp}$ and the determinant in \eqref{topodec} is computed on~$H$. The auxiliary tau functions in~\eqref{topodec} can be easily found to be $t^{\frac12\operatorname{Tr}\big(\mathfrak S^2-\Theta_0^2-\Theta_t^2\big)}$ and $1$ (here we ignore constant factors that can be chosen arbitrarily).

Consequently, whenever the inverse monodromy problem for the auxiliary $3$-point Fuchsian systems can be solved, the $4$-point tau function admits an explicit Fredholm determinant representation. The crucial point for us is that, when 1 of 3 singular points is semi-degenerate (which in the FST$_N$ case is true for both $\Phi_+$ and $\Phi_-$), the inverse monodromy problem does have an explicit solution in terms of generalized hypergeometric functions $_NF_{N-1}$. Here is its construction.

Consider a semi-degenerate Fuchsian system with $3$ singular points $z_0=0$, $z_1=1$, $z_2=\infty$. With a slight abuse of notation, we write the corresponding connection as $\partial_z-A(z)$ and denote by $ A_0$, $ A_1$, $ A_\infty=- A_0- A_1$ its residues at the poles. Assume that $ A_0$, $ A_1$, $ A_\infty$ can be represented as
\begin{gather*} 
A_{\nu} = G_{\nu}^{-1} \Theta_{\nu} G_{\nu},\qquad \nu=0,1,\infty,
\end{gather*}
with $G_0,G_1\in \mathrm{GL}(N,\mathbb{C})$, $G_{\infty}=\mathds 1$ and diagonal $\Theta_0$, $\Theta_1$, $\Theta_\infty$:
\begin{gather*}
\Theta_{\nu}=\operatorname{diag}\big(\theta^{(1)}_{\nu},\ldots,\theta^{(N)}_{\nu}\big),\qquad \nu=0,\infty,\qquad \Theta_1=\Lambda \cdot \operatorname{diag}(1,0,\ldots,0).
\end{gather*}
It is also assumed that all the eigenvalues of $\Theta_0$ and $\Theta_\infty$ are distinct. Such data correspond to a rigid local system and the matrix elements of $ A_1$ are determined by $\Theta_0$, $\Theta_{\infty}$ almost uniquely:
\begin{gather}\label{ME_Ajm}
\big(R^{-1} A_1 R\big)_{jm}=-\prod_{k}\big(\theta^{(j)}_\infty+\theta^{(k)}_0\big) \prod_{k\ne m}\big(\theta^{(m)}_\infty-\theta^{(k)}_\infty\big)^{-1}.
\end{gather}
Here $R=\operatorname{diag}\big(r^{(1)},\ldots,r^{(N)}\big)$ with all $r^{(k)}\in\Cb^*$ appears due to the freedom of conjugation of~$ A_{0,1,\infty}$ by a matrix preserving the diagonal gauge at~$\infty$. The matrix~$A_1$ has rank 1 and can therefore be written as $ A_1=-Ru^T\otimes vR^{-1}$, where the components of the row vectors~$u$ and~$v$ are given by
\begin{gather}\label{uvdeff}
u^{(j)}= \prod_{k}\big(\theta^{(j)}_\infty+\theta^{(k)}_0\big), \qquad v^{(m)} = \prod_{k\ne m}\big( \theta^{(m)}_\infty-\theta^{(k)}_\infty\big)^{-1}.
\end{gather}
In what follows, we will need an explicit expression for the diagonalizing transformation $G_0=\bar R G_{\Theta_\infty,\Theta_0} R^{-1}$, where $\bar R=\operatorname{diag}\big(\bar r^{(1)},\ldots,\bar r^{(N)}\big)$ and all $\bar r^{(k)}\in\Cb^*$ can be fixed arbitrarily, together with some related formulas, cf.\ Lemma~\ref{lemrig}:
\begin{gather}
( G_{\Theta_\infty,\Theta_0})_{kl}= \prod_{s\ne l}\frac{\theta^{(s)}_\infty+\theta^{(k)}_0}{\theta^{(s)}_\infty-\theta^{(l)}_\infty} , \qquad
\big( G_{\Theta_\infty,\Theta_0}^{-1}\big)_{kl}= \prod_{s\ne l}\frac{\theta^{(s)}_0+\theta^{(k)}_\infty}{\theta^{(s)}_0-\theta^{(l)}_0} ,\nonumber\\
\big(G_{\Theta_\infty,\Theta_0}u^T\big)_j=(-1)^{N-1} \prod_{k}\big(\theta^{(j)}_0+\theta^{(k)}_\infty\big), \qquad
\big(v G_{\Theta_\infty,\Theta_0}^{-1}\big)_m = \prod_{k\ne m}\big( \theta^{(k)}_0-\theta^{(m)}_0\big)^{-1}.\label{exprG0}
\end{gather}

\begin{Lemma}\label{lemmahyp}Let
\begin{gather}\label{3pt}
\partial_z\Phi(z) = \Phi(z) A(z),\qquad A(z) =\frac{ A_0}{z}+\frac{A_1}{z-1},
\end{gather}
be a Fuchsian system with $ A_1$ given by \eqref{ME_Ajm} and $A_0+ A_1+\Theta_{\infty}=0$.
\begin{enumerate}\itemsep=0pt
\item[$1.$] The system \eqref{3pt} has a unique solution $\Phi^{(\infty)}(z)$ with the asymptotics $\Phi^{(\infty)}(z\to\infty)\simeq(-z)^{-\Theta_\infty}\big(\mathds 1+O\big(z^{-1}\big)\big)$. It is explicitly given by
\begin{gather*}
\Phi^{(\infty)}(z)=R(-z)^{-\Theta_\infty} \Psi_{\Theta_{\infty},\Theta_0}(1/z) R^{-1},
\end{gather*}
with
\begin{gather}
\big(\Psi_{\Theta_{\infty},\Theta_0}(1/z)\big)_{jm}\nonumber\\
\qquad{} = N^{(\infty)}_{jm} z^{\delta_{jm}-1}
{}_NF_{N-1}\left(\genfrac{}{}{0pt}{0}{\bigl\{1-\delta_{jm}+\theta^{(k)}_0+\theta^{(j)}_\infty\bigr\}_{k=\overline{1,N}} }
{\bigl\{1+\theta^{(j)}_\infty-\theta^{(k)}_\infty
+\delta_{mk}-\delta_{jm}\bigr\}_{k=\overline{1,N};k\ne j}}\Bigg|\,\frac{1}{z}\right),\label{hyperg}
\end{gather}
where $_NF_{N-1}$ is the generalized hypergeometric function, $j,m=1,\ldots,N$, and
\begin{gather*}
N^{(\infty)}_{jm}=\begin{cases}\displaystyle\frac{\prod\limits_{k}\big(\theta^{(j)}_\infty+\theta^{(k)}_0\big) \prod\limits_{k\ne m}\big(\theta^{(m)}_\infty-\theta^{(k)}_\infty\big)^{-1}}{1+\theta^{(j)}_\infty-\theta^{(m)}_\infty},& j\ne m,\\
1, & j=m.
\end{cases}
\end{gather*}
\item[$2.$] Similarly, there exists a unique solution $\Phi^{(0)}(z)$ of \eqref{3pt} with the normalized asymptotics $\Phi^{(0)}(z\to 0) \simeq (-z)^{\Theta_0}(\mathds 1+O(z)) G_0$. It can be written as
\begin{gather*}
\Phi^{(0)}(z) = \bar R(-z)^{\Theta_0}\Psi_{\Theta_0,\Theta_{\infty}}(z) G_{\Theta_\infty,\Theta_0} {R}^{-1},
\end{gather*}
where $\Psi_{\Theta_0,\Theta_{\infty}}(z)$ is obtained from the previous formulas by the exchange $\Theta_0\leftrightarrow\Theta_{\infty}$, \smash{$z\leftrightarrow z^{-1}$}, just as the notation suggests.
\item[$3.$] The solutions $\Phi^{(\infty)}(z)$ and $\Phi^{(0)}(z)$ are related by $\Phi^{(0)}(z)=\bar R S_{\Theta_0,\Theta_\infty} R^{-1}\Phi^{(\infty)}(z)$, where the matrix elements of $S_{\Theta_0,\Theta_\infty}$ are given by
\begin{gather*}
\big( S_{\Theta_0,\Theta_\infty}\big)_{lj}=
\prod_{k\ne l} \frac{\Gamma\big( 1+\theta^{(l)}_0-\theta^{(k)}_0\big)}{\Gamma\big(1 - \theta^{(j)}_\infty -\theta^{(k)}_0 \big)}
\prod_{k\ne j} \frac{\Gamma\big( \theta^{(k)}_\infty-\theta^{(j)}_\infty \big)}{\Gamma\big(\theta^{(k)}_\infty+\theta^{(l)}_0 \big)}.
\end{gather*}
The inverse transformation is obtained from $(S_{\Theta_0,\Theta_\infty})^{-1}=S_{\Theta_\infty,\Theta_0}$.
\item[$4.$] The inverse matrix $\Psi_{\Theta_{\infty},\Theta_0}(z)$ can be written as
\begin{gather}\label{invps}
\Psi_{\Theta_{\infty},\Theta_0}(z)^{-1}=D_{\Theta_{\infty},\Theta_0}\Psi_{-\Theta_{\infty},-\Theta_0}(z)^T D_{\Theta_{\infty},\Theta_0}^{-1},
\end{gather}
with $D_{\Theta_{\infty},\Theta_0}=\operatorname{diag}\big\{\frac{u^{(1)}}{v^{(1)}},\ldots,\frac{u^{(N)}}{v^{(N)}}\big\}$.
\end{enumerate}
\end{Lemma}

We are now prepared to formulate the main result of this section, expressing the generic FST$_N$ tau function in terms of semi-degenerate monodromy of the associated Fuchsian system. The monodromy will be parameterized in the same as above in \eqref{Mt1}--\eqref{M0Mt}, \eqref{skldef} and will satisfy the same genericity conditions. In particular, $M_0$, $M_{\infty}$ and $M_0 M_t$ are assumed diagonalizable.
\begin{Theorem} \label{theofstfr}Let the pairs $\Theta_0$, $-\mathfrak S$ and $\mathfrak S$, $\Theta_\infty$ have disjoint simple non-resonant spectra. The corresponding FST$_N$ tau function admits a Fredholm determinant representation,
\begin{gather}\label{taufstfr}
\tau_{\mathrm{FST}}(t)=\operatorname{const}\cdot t^{\frac12\operatorname{Tr}\big(\mathfrak S^2-\Theta_0^2-\Theta_t^2\big)}
\det \left(\begin{matrix}
\mathds 1 & \mathsf a_{+-}\\
\mathsf a_{-+}
& \mathds 1
\end{matrix}\right),
\end{gather}
where $\mathsf a_{\pm\mp}$ are integral operators defined by \eqref{aops}, with
\begin{gather}\label{psipmfst}
\Psi_+(z)=\Psi_{\mathfrak S,\Theta_\infty}(z),\qquad \Psi_-(z)=D^{-1} t^{-\mathfrak S}\Psi_{-\mathfrak S,\Theta_0}\big(\tfrac tz\big)
\end{gather}
and $D=\operatorname{diag}\big(\exp\tilde\beta^{(1)},\ldots,\exp\tilde\beta^{(N)}\big)$ with $\sum_k\tilde\beta^{(k)}=0$. The conjugacy class of monodromy contains the following representative, expressed in terms of $\mathfrak S$, $D$:
\begin{gather*}
 M_0=D^{-1}S_{-\mathfrak S,\Theta_0}{\rm e}^{2\pi{\rm i} \Theta_0}S_{\Theta_0,-\mathfrak S}D,\qquad M_0M_t=( M_1 M_\infty)^{-1}={\rm e}^{2\pi{\rm i} \mathfrak S},\\
 M_{\infty}=S_{\mathfrak S,\Theta_\infty}{\rm e}^{2\pi{\rm i} \Theta_\infty}S_{\Theta_\infty,\mathfrak S}.
\end{gather*}
\end{Theorem}
\begin{Remark}
The last theorem extends Theorem~A of \cite{GL16} to higher rank $N\ge2$. Note, however, that in \cite{GL16} all monodromies are assumed to have unit determinant, whereas here we place ourselves in the gauge where the rank of $M_{t,1}-\mathds 1$ is~1. This leads to an additional elementary factor of the form $(1-t)^{\alpha}$ (called the $\mathrm{U}(1)$-factor in the gauge theory context) which should be taken into account when comparing both results for $N=2$.
\end{Remark}

\section{Series representation\label{secseries}}
The series representation of $\tau_{\mathrm{FST}}(t)$ is obtained by expanding the determinant \eqref{taufstfr} into a sum of principal minors. The basic building blocks of this construction are given by the coefficients of the Fourier expansion of the kernel
\begin{gather}\check{\mathsf a}(z,z') \equiv \check{\mathsf a}( z,z'\,|\,\Theta_0,\Theta_\infty)=
\frac{\Psi_{\Theta_{\infty},\Theta_0}\big(z^{-1}\big)\Psi_{\Theta_{\infty},\Theta_0}\big( z'^{-1}\big)^{-1}-\mathds 1}{z-z'}\nonumber\\
\hphantom{\check{\mathsf a}(z,z')}{} =\sum_{p,q\in\mathbb Z'_+}\check{\mathsf a}_{-q,p} z^{-q-\frac12}z'^{-p-\frac12},\label{apq1}
\end{gather}
where $\mathbb Z'=\mathbb Z+\frac12$ and $\mathbb Z'_{\pm}=\mathbb Z'_{\gtrless 0}$. The Fourier modes $\check{\mathsf a}_{-q,p}=\check{\mathsf a}_{-q,p}(\Theta_0,\Theta_\infty)$ are themselves $N\times N$ matrices whose entries will be denoted by $\check{\mathsf a}_{-q,k;p,l}$, with $k,l=1,\ldots, N$.
\begin{Lemma}\label{lemacheck}We have
\begin{subequations}
\begin{gather}\label{lemapq1}
\check{\mathsf a}_{-q,k;p,l}= \frac{\varphi_{-q}^{(k)} \bar\varphi_{p}^{(l)}}{p+q+\theta_{\infty}^{(k)}-\theta_{\infty}^{(l)}},
\end{gather}
where $\varphi_{-q}^{(k)}=\varphi_{-q}^{(k)}(\Theta_0,\Theta_\infty)$ and $ \bar\varphi_{p}^{(l)}= \bar\varphi_{p}^{(l)}(\Theta_0,\Theta_\infty)$ are given by
\begin{gather}
\varphi_{-q}^{(k)}=\frac{\prod\limits_{m=1}^N\big(\theta_{\infty}^{(k)}+\theta_0^{(m)}\big)_{q+\frac12}}{\big(q-\frac12\big)! \prod\limits_{m\ne k}^N\big(1+\theta_{\infty}^{(k)}-\theta_{\infty}^{(m)}\big)_{q-\frac12}},\nonumber\\
\bar\varphi_{p}^{(l)}=\frac{(-1)^{N} \prod\limits_{m=1}^N\big(1-\theta_{\infty}^{(l)}-\theta_0^{(m)}\big)_{p-\frac12}}{\big(p-\frac12\big)! \prod\limits_{m\ne l}^N\big(\theta_{\infty}^{(m)}-\theta_{\infty}^{(l)}\big)_{p+\frac12}}.\label{lemapq2}
\end{gather}
\end{subequations}
\end{Lemma}
\begin{proof}
From the differential equation \eqref{3pt} for $\Phi^{(\infty)}(z)$ it follows that $\Psi_{\Theta_{\infty},\Theta_{0}}\big(z^{-1}\big)$ in its turn satisfies
\begin{gather}\label{apq2}
z\partial_z \Psi_{\Theta_{\infty},\Theta_{0}}\big(z^{-1}\big)=\big[\Theta_{\infty},\Psi_{\Theta_{\infty},\Theta_{0}}\big(z^{-1}\big)\big]+\Psi_{\Theta_{\infty},\Theta_{0}}\big(z^{-1}\big) \frac{R^{-1} A_1 R}{z-1}.
\end{gather}
Applying to the decomposition \eqref{apq1} the operator $z\partial_z+z'\partial_{z'}+1$ and using \eqref{apq2}, we get
\begin{gather}\label{apq3}
\sum_{p,q\in\mathbb Z'_+}\big((p+q)\check{\mathsf a}_{-q,p}+\big[\Theta_\infty,\check{\mathsf a}_{-q,p}\big]\big) z^{-q-\frac12}z'^{-p-\frac12}=\varphi(z)\otimes \bar\varphi(z'),
\end{gather}
where $u$, $v$ are defined by \eqref{uvdeff} and
\begin{gather*}
\varphi(z)=\frac{\Psi_{\Theta_{\infty},\Theta_{0}}\big(z^{-1}\big)}{z-1} u^T,\qquad \bar\varphi(z)=-v \frac{\Psi_{\Theta_{\infty},\Theta_{0}}\big(z^{-1}\big)^{-1}}{z-1}.
\end{gather*}
The factorization of the right hand side of \eqref{apq3} implies that the integral operator with the kernel $\check{\mathsf a}(z,z')$ becomes a Cauchy type matrix in the Fourier basis. Namely, the decompositions
\begin{gather*}
\varphi(z)=\sum_{q\in\mathbb Z'_+}\varphi_{-q} z^{-q-\frac12}, \qquad \bar\varphi(z)=\sum_{p\in\mathbb Z'_+}\bar\varphi_{p} z^{-p-\frac12},
\end{gather*}
yield the structure \eqref{lemapq1}. For the computation of Fourier components of $\varphi(z)$, $\bar\varphi(z)$ observe that, taking the diagonal elements of the relation \eqref{apq2}, one obtains
\begin{gather*}
\big( z\partial_z \Psi_{\Theta_{\infty},\Theta_{0}}\big(z^{-1}\big)\big)_{kk}=-\varphi^{(k)}(z) v^{(k)}.
\end{gather*}
Substituting therein the hypergeometric expressions~\eqref{hyperg}, we immediately get the first of relations~\eqref{lemapq2}. The explicit form of $\bar\varphi_p$ may be extracted in a similar manner from the diagonal of a counterpart of the equation~\eqref{apq2} satisfied by $\Psi_{\Theta_{\infty},\Theta_{0}}\big(z^{-1}\big)^{-1}$. Recall that this inverse can be expressed by means of~\eqref{invps}.
\end{proof}

The integral kernels $\mathsf a_{\pm\mp}(z,z')$ of the operators appearing in the Fredholm determinant \eqref{taufstfr} have the structure analogous to \eqref{apq1}. Decomposing them as
\begin{gather*}
\mathsf a_{\pm\mp}(z,z')=\sum_{p,q\in\mathbb Z'_+} \mathsf a_{\pm q,\mp p} z^{\pm q-\frac12}z'^{\pm p-\frac12},
\end{gather*}
it follows from Theorem~\ref{theofstfr} (cf.\ equation~\eqref{psipmfst}) that
\begin{gather*}
\mathsf a_{-q,k;p,l}= {\rm e}^{\tilde\beta^{(l)}-\tilde\beta^{(k)}}t^{\sigma^{(l)}-\sigma^{(k)}+p+q}\check{\mathsf a}_{-q,k;p,l}(\Theta_0,-\mathfrak S),\\
\mathsf a_{p,l;-q,k}= \check{\mathsf a}_{-p,l;q,k}(\Theta_{\infty},\mathfrak S),
\end{gather*}
where $\check{\mathsf a}$'s are explicitly given by Lemma~\ref{lemacheck}.

The principal minor expansion of the determinant \eqref{taufstfr} can now be calculated following the scheme outlined in \cite[Section~2.2.1]{CGL}. We obtain
\begin{gather}
\tau_{\mathrm{FST}}(t)=t^{\frac12\big(\boldsymbol{\sigma}^2-\boldsymbol{\theta}_0^2-\boldsymbol{\theta}_t^2\big)}\sum_{\boldsymbol{\mathsf p},\boldsymbol{\mathsf q}\colon |\boldsymbol{\mathsf p}|=|\boldsymbol{\mathsf q}|}{\rm e}^{\boldsymbol{\tilde\beta}\cdot\boldsymbol{w}}t^{\boldsymbol{\sigma}\cdot\boldsymbol{w}+\sum\limits_{p\in\boldsymbol{\mathsf p}}p+\sum\limits_{q\in\boldsymbol{\mathsf q}}q}\nonumber \\
\hphantom{\tau_{\mathrm{FST}}(t)=}{} \times (-1)^{|\boldsymbol{\mathsf p}|} \det \check{\mathsf a}_{-\boldsymbol{\mathsf q},\boldsymbol{\mathsf p}}(\Theta_0,-\mathfrak S)
\det \check{\mathsf a}_{-\boldsymbol{\mathsf p},\boldsymbol{\mathsf q}}(\Theta_{\infty},\mathfrak S).\label{sertaufst}
\end{gather}
The notation is as follows:
\begin{itemize}\itemsep=0pt
\item $\boldsymbol{\mathsf p}=\mathsf p^{(1)}\sqcup\cdots\sqcup \mathsf p^{(N)}$, $\boldsymbol{\mathsf q}=\mathsf q^{(1)}\sqcup\cdots \sqcup\mathsf q^{(N)}$, where each of $\mathsf p^{(k)}$, $\mathsf q^{(k)}$ is a finite subset of $\mathbb Z'_+$. The elements of $\mathsf p^{(k)}$ and $-\mathsf q^{(k)}$ may be considered as positions of particles and holes in a~Maya diagram of color $k\in \{1,\ldots, N \}$. The summation in~\eqref{sertaufst} is carried over $N$-tuples of Maya diagrams satisfying the condition of global balance: $\sharp(\mathrm{particles})=\sharp(\mathrm{holes})$.
\item $\boldsymbol{w}=\big( w^{(1)},\ldots,w^{(N)}\big)$, where $w^{(k)}=\big|\mathsf p^{(k)}\big|-\big|\mathsf q^{(k)}\big|$ is the charge of the appropriate Maya diagram. Obviously, we have $\sum\limits_{k=1}^N w^{(k)}=0$, i.e., $\boldsymbol{w}$ belongs to the $\mathfrak{sl}_N$ root lattice, to be denoted by $\mathfrak Q_{N-1}$. We denote $\boldsymbol{\sigma}= \big(\sigma^{(1)},\ldots,\sigma^{(N)}\big)$, $\boldsymbol{\tilde\beta}=\big(\tilde\beta^{(1)},\ldots,\tilde\beta^{(N)}\big)$ etc., so that, for instance, $\boldsymbol{\sigma}\cdot\boldsymbol{w}=\sum\limits_{k=1}^N w^{(k)}\sigma^{(k)}$ and $\operatorname{Tr}\mathfrak S^2=\boldsymbol{\sigma}^2$.
\item $\check{\mathsf a}_{-\boldsymbol{\mathsf q},\boldsymbol{\mathsf p}}$ are $|\boldsymbol{\mathsf p}|\times |\boldsymbol{\mathsf p}|$ matrices obtained by restriction of $\check{\mathsf a}$ written in the Fourier basis to rows and columns labeled by $-\boldsymbol{\mathsf q}$ and $\boldsymbol{\mathsf p}$.
\item Using the correspondence between Maya diagrams and charged Young diagrams, the sum in~\eqref{sertaufst} can be alternatively rewritten as $\sum\limits_{\boldsymbol{w}\in\mathfrak Q_{N-1}}\sum\limits_{\bs Y\in\mathbb Y^N}$, where $\mathbb Y$ denotes the set of all Young diagrams. The reader is referred to~\cite{CGL,GL16,GL17} for the details.
\end{itemize}
Let us stress that the structure of \eqref{sertaufst} is not specific for semi-degenerate monodromy. The spectral constraints played a role above only in finding the explicit form of the 3-point solu\-tions~$\Psi_{\pm}(z)$ and matrix elements $\check{\mathsf a}_{-q,p}$ in Lemma~\ref{lemacheck}. An additional bonus is the Cauchy structure~\eqref{lemapq1} of these matrix elements which enables one to calculate the determinants in~\eqref{sertaufst} in a factorized form. In the generic situation, one has a sum of $N-1$ Cauchy matrices instead of just~1.

\section[Scalar $_N F_{N-1}$ kernel]{Scalar $\boldsymbol{_N F_{N-1}}$ kernel} \label{sec_fred}

\looseness=-1 It was shown in~\cite{GL16} that, for special monodromy data, the tau function of the Painlev\'e~VI equation (FST$_2$ system) can be expressed as a \textit{scalar} Fredholm determinant on the interval $\mathcal B:=[0,t]\subset \mathbb R$. The relevant integrable ${}_2F_1$ kernel is related to the so-called ZW-measures~\cite{BO}. Here we generalize the former result to $N>2$ by constructing a class of FST$_N$ tau functions which can be represented by scalar Fredholm determinants with an integrable ${}_NF_{N-1}$ kernel on~$\mathcal B$.

We denote $M_0M_t={\rm e}^{2\pi{\rm i} \mathfrak S}$, $\mathfrak S=\operatorname{diag}\big\{\sigma^{(1)},\ldots,\sigma^{(N)}\big\}$ as before, and consider a special monodromy for the inner auxiliary 3-point problem with singular points $0$, $t$, $\infty$:
\begin{gather}
M_t=1+{\rm e}^{-{\rm i}\pi {\mf S}} u^T \otimes v {\rm e}^{{\rm i}\pi {\mf S}} .\label{eq:monodromy:nilpotent}
\end{gather}
Here the row vectors $u$ and $v$ are defined by
\begin{gather}\label{spec_uv}
u^T=2\pi{\rm i}\sum_{i\in I}s_i e_i,\qquad v=\sum_{j\in J}s_j^{-1} e^j ,
\end{gather}
and $I$, $J$ are two non-empty non-intersecting subsets of all indices:
\begin{gather}\label{IJ}
I, J\subset \{1,\ldots,N\} ,\qquad I\cap J=\varnothing.
\end{gather}
The parameters $s_i$ are non-zero complex numbers and $e_i$ are the standard basis vectors in~${\mathbb C}^N$. Due to~(\ref{spec_uv}) and~(\ref{IJ}), one has $v\cdot u^T=0$. From the explicit expression for $M_t$ it follows that the spectra of $M_0$ and $M_{0}M_t$ coincide.

The monodromy matrix $M_t$ is non-diagonalizable. Therefore, Theorem~\ref{theofstfr} together with the previous parameterization of 3-point Fuchsian systems are not directly applicable. Nevertheless it is not difficult to guess the form of a 3-point Fuchsian system having the prescribed monodromy~(\ref{eq:monodromy:nilpotent}):
\begin{gather}
\frac{{\rm d}\Phi_-(z)}{{\rm d}z}=\Phi_-(z)\left[\frac{\mathfrak S}{z}+ \frac{u^T\otimes v }{2\pi{\rm i}(z-t)}\right].\label{FscF}
\end{gather}
Rewriting this equation in terms of a new function $\chi(z)=\Phi_-(tz) (-z)^{-\mf S}$, one obtains
\begin{gather*}
\frac{{\rm d}\chi(z)}{{\rm d}z}=\chi(z)\frac{ (-z)^{\mf S} u^T\otimes v (-z)^{-\mf S}}{2\pi{\rm i}(z-1)} .
\end{gather*}

Its solution (normalized at $z=\infty$) can be written as
\begin{gather*}
\chi(z)=\mathds 1+\sum_{i\in I,\, j\in J}\frac{ (-z)^{\sigma^{(i)}-\sigma^{(j)}} l_{\sigma^{(j)}-\sigma^{(i)}}(1/z)}{\sigma^{(i)}-\sigma^{(j)}} s_i s_j^{-1} e_i\otimes e^j ,
\end{gather*}
where
\begin{gather}\label{lfunc}
l_\sigma(z)={}_2F_1\left(\left.\begin{matrix}
1,\sigma \\ 1+\sigma\end{matrix}\right|z\right).
\end{gather}
This function is naturally defined on $\Cb\backslash\mathbb R_{\ge1}$. We are going to use a jump property of the rescaled function $l_\sigma\big(\frac tz\big)$ which can be derived from~(\ref{lfunc}):
\begin{gather}\label{l_jump}
l_\sigma\left(\frac{t}{x+{\rm i}0}\right) - l_\sigma\left(\frac{t}{x-{\rm i}0}\right)=-2\pi{\rm i} \sigma \left(\frac{x}{t}\right) ^\sigma,\qquad 0<x<t .
\end{gather}
The structure of analytic continuation of $\chi(z)$ follows from the analytic properties of $l_\sigma(z)$. For example, the continuation of $\chi(z)$ along the curve $\gamma_1$ encircling $z=1$ in the anti-clockwise direction gives
\begin{gather*}
\chi(\gamma_1.z)=\chi(z)+{\rm e}^{-{\rm i}\pi {\mf S}} u^T \otimes v {\rm e}^{{\rm i}\pi {\mf S}} =\big(\mathds 1+{\rm e}^{-{\rm i}\pi {\mf S}} u^T \otimes v {\rm e}^{{\rm i}\pi {\mf S}}\big)\chi(z),
\end{gather*}
where we used the jump relation (\ref{l_jump}) for $t=1$ and the nilpotency properties. Therefore, the monodromy of $\Phi_-(z)=\chi\big(\frac zt\big) \big({-}\frac zt\big)^{\mf S}$ is given by
\begin{gather*}
\Phi_-(\gamma_t.z)=\big(\mathds 1+{\rm e}^{-{\rm i}\pi {\mf S}} u^T \otimes v {\rm e}^{{\rm i}\pi {\mf S}} \big)\Phi_-(z),\qquad \Phi_-\big(\gamma_\infty^{-1}.z\big)={\rm e}^{2\pi{\rm i}\mf S}\Phi_-(z).
\end{gather*}
Thus $\Phi_-(z)$ solves the Fuchsian system \eqref{FscF} and has prescribed monodromy. Below we also need the explicit expression for $\Psi_-(z)=(-z)^{-\mf S}\Phi_-(z)$:
\begin{gather}\label{Psim_series}
\Psi_-(z)=t^{-\mf S}\left({\mathds 1}+\sum_{{i\in I,\, j\in J}}\frac{l_{\sigma^{(j)}-\sigma^{(i)}}\big(\frac tz\big)}{\sigma^{(i)}-\sigma^{(j)}} s_is_j^{-1} e_i\otimes e^j\right).
\end{gather}

Let us investigate the integral kernels of the operators ${\sf a_{-+}}$ and ${\sf a_{+-}}$ appearing in the Fredholm determinant representation of the tau function, $\tau(t)=\det(\mathds{1}-{\sf a_{+-}} {\sf a_{-+}})$,
and given by~\eqref{aops}. In the initial setting, ${\sf a_{\pm\mp}}$ acted on functions on a circle $\mathcal C$ centered at $z=0$ and having radius $R\in(t,1)$. Let us now shrink the contour $\mathcal C$ to the branch cut
$\mathcal B$ of $\Psi_-(z)$, and in this way transform $L^2\big(S^1,\mathbb{C}^N\big)\leadsto \mathcal{W}=L^2\big(\mathcal B,\mathbb{C}^N\big)_{\mathrm{up}}\oplus L^2\big(\mathcal B,\mathbb{C}^N\big)_{\mathrm{down}}=\mathcal{W}_{s}\oplus \mathcal{W}_{a}$, where $s$ and $a$ stand for symmetric and antisymmetric functions with respect to exchange of their boundary values on the two sides of $\mathcal{B}$:
\begin{gather*}
\mathcal W_{s}= \{f\in\mathcal W\colon f(x+{\rm i}0)=+ f(x-{\rm i}0),\,x\in\mathcal B \} ,\\
\mathcal W_{a}= \{f\in\mathcal W\colon f(x+{\rm i}0)=- f(x-{\rm i}0),\,x\in\mathcal B \} .
\end{gather*}
Since $\Psi_+(z)$ has no jump on $\mathcal B$ symmetric functions contain the image and belong to the kernel of $\sf a_{+-}$:
\begin{gather*}
({\sf a}_{+-})_{ss}=({\sf a}_{+-})_{as}=({\sf a}_{+-})_{aa}=0.
\end{gather*}
Consequently, the only part of $\sf a_{-+}$ important for the determinant computation is $({\sf a_{-+}})_{as}$.

Using the evaluation (\ref{Psim_series}) and the jump property~(\ref{l_jump}) of $l_\sigma\big(\frac tz\big)$, it is straightforward to obtain formulas for the jumps of $\Psi_-(z)$ and $\Psi_-^{-1}(z)$ on~$\mathcal B$. For $x,y\in(0,t)$, one has
\begin{gather}
\Psi_-(x+{\rm i}0)-\Psi_-(x-{\rm i}0)= x^{-\mf{S}} u^T \otimes v x^{\mf{S}} t^{-\mf{S}},\nonumber\\
\label{Psiinv_jump}
\Psi_-^{-1}(y+{\rm i}0)-\Psi_-^{-1}(y-{\rm i}0)= - t^{\mf{S}} y^{-\mf{S}} u^T \otimes v y^{\mf{S}} .
\end{gather}

\begin{Lemma}\label{lemmaF21}Let $\max_{i,j}\big|\Re \big(\sigma^{(i)}-\sigma^{(j)}\big)\big|<1$. For $n\in\mathbb Z_{\ge0}$, define $X_n:=\operatorname{Tr}(\mathsf a_{+-}\mathsf a_{-+})^n$. We have the equality
\begin{gather}\label{Kn}
X_n=\operatorname{Tr} K^n,
\end{gather}
where $ K$ is an integral operator on $L^2(\mathcal B)$ with the kernel
\begin{gather}\label{kernelKF}
K(x,y) =-\sum_{i\in I,\,j\in J} s_j^{-1} x^{\sigma^{(j)}} \frac{\big(\Psi_+(x)\Psi_+(y)^{-1}\big)_{ji}}{x-y} y^{-\sigma^{(i)}} s_i .
\end{gather}
\end{Lemma}

\begin{proof} Let $g\in L^2\big(\mathcal B,\mathbb{C}^N\big)$. The action of $\mathsf a_{-+}$ on $f=g\oplus g\in\mathcal W_s$ reads
\begin{gather*}
(\mathsf a_{-+} f)(z)= \frac{1}{2\pi{\rm i}} \int_0^t \frac{\Psi_-(z)\big[ \Psi_-^{-1}(y-{\rm i}0)-\Psi_-^{-1}(y+{\rm i}0)\big]}{z-y} g(y) {\rm d}y \\
\hphantom{(\mathsf a_{-+} f)(z)}{} = \frac{1}{2\pi{\rm i}} \int_0^t \frac{ y^{-\mf{S}} u^T \otimes v y^{\mf{S}} }{z-y} g(y) {\rm d}y ,\qquad z\not\in \mathcal{B}\,,
\end{gather*}
where at the second step we used the relation \eqref{Psiinv_jump} and nilpotency. The next step is to compute the projection~$\Pi_a$ of this expression on~$\mathcal W_a$. Write the result as $\Pi_a\circ \mathsf a_{-+} f=h\oplus(-h)$, with some $h\in L^2\big(\mathcal B,\mathbb{C}^N\big)$. For $x\in \mathcal B$, one has
\begin{gather}
 h(x)=\frac12\left[(\mathsf a_{-+} f)(x+{\rm i}0)-(\mathsf a_{-+} f)(x-{\rm i}0)\right] \nonumber\\
\hphantom{h(x)}{} =-\frac{1}{4\pi i} \int_{0}^{t} y^{-\mf{S}} u^T \otimes v y^{\mf{S}} \left(\frac{1}{y-x-{\rm i}0}-\frac{1}{y-x+{\rm i}0} \right) g(y) {\rm d}y \nonumber\\
\label{aauu} \hphantom{h(x)}{}=-\frac{1}{2} x^{-\mf{S}} u^T \otimes v x^{\mf{S}} g(x) .
\end{gather}
At last, write $\mathsf a_{+-}\circ\Pi_a\circ \mathsf a_{-+} f$ as $\tilde g\oplus \tilde g\in\mathcal W_s$. The expression~\eqref{aauu} for $h(x)$ implies that
\begin{gather*}
\tilde g(x)=\frac{1}{2\pi{\rm i}} \int_0^t \mathsf a_{+-}(x,y) y^{-\mf{S}} u^T \otimes v y^{\mf{S}} g(y) {\rm d}y,\qquad x\in\mathcal B.
\end{gather*}
The last formula describes the action of $\mathsf a_{+-}\circ\Pi_a\circ \mathsf a_{-+}$ on $\mathcal W_s$. To obtain \eqref{Kn}--\eqref{kernelKF}, it now suffices to compute the $n$th power of this operator, use the cyclic property of the trace and the explicit form~\eqref{aops} of the integral kernel $\mathsf a_{+-}(x,y)$.
\end{proof}

We now summarize the developments of this section:
\begin{Theorem}\label{theorem2F1} Let $I$, $J$ be two non-intersecting subsets of $\{1,\ldots,N\}$ and let
\begin{gather*}
\bigl\{\theta_{\infty}^{(k)}\bigr\}_{k=\overline{1,N}},\qquad \bigl\{\sigma^{(k)}\bigr\}_{k=\overline{1,N}},\qquad \{s_i \}_{i\in I\cup J},
\end{gather*}
be complex parameters such that
\begin{itemize}\itemsep=0pt
\item $s_i\ne 0$ for $i\in I\cup J$ and $\sigma^{(i)}\ne \sigma^{(j)}$ for $i\ne j$;
\item $\max_{i,j\in\{1,\ldots,N\}}\big|\Re \big(\sigma^{(i)}-\sigma^{(j)}\big)\big|<1$.
\end{itemize}
Let $K$ be the integral operator on $L^2([0,t])$ with the integrable kernel \eqref{kernelKF}, where
\begin{subequations}
\label{ac618}
\begin{gather}
(\Psi_+(x)) _{jm} = C_{jm} x^{1-\delta_{jm}}{}_{N}F_{N-1} \left(\left.\genfrac{}{}{0pt}{}{\bigl\{1-\delta_{jm}+\theta^{(k)}_{\infty}+\sigma^{(j)}\bigr\}_{k=\overline{1,N}}}{
\bigl\{1+\sigma^{(j)}-\sigma^{(k)}+\delta_{mk}-\delta_{jm}\bigr\}_{k=\overline{1,N},k\neq j}}\right|x\right) ,\\
\label{invhyp}
\bigl(\Psi_+( y)^{-1}\bigr)_{mi}= C'_{mi} y^{1-\delta_{im}} {}_{N}F_{N-1} \left(\left.\genfrac{}{}{0pt}{}{\bigl\{1-\delta_{im}-\theta^{(k)}_{\infty}-\sigma^{(i)}\bigr\}_{k=\overline{1,N}}}{
\bigl\{1-\sigma^{(i)}+\sigma^{(k)}+\delta_{mk}-\delta_{im}\bigr\}_{k=\overline{1,N},k\neq i}}\right|y\right) ,\\
C_{jm}=\frac{\prod\limits_{k=1}^N\big(\sigma^{(j)}+\theta^{(k)}_{\infty}\big)}{\big(\sigma^{(j)}-\sigma^{(m)}+1\big)\prod\limits_{k\neq m}\big(\sigma^{(m)}-\sigma^{(k)}\big)}=-C'_{mj}\qquad\text{for } j\neq m,\nonumber\\ C_{mm}=C'_{mm}=1.
\end{gather}
\end{subequations}
Then the scalar Fredholm determinant $\tau(t)=\det (\mathds 1 - K_{[0,t]})$ is a tau function of the FST$_N$ system with $\Lambda_t=0$.
\end{Theorem}

\begin{Remark} Note that since $\Psi_+(x)\Psi_+(x)^{-1}=\mathds 1$, the kernel $K(x,y)$ is not singular along the diagonal $x=y$. The formulas~\eqref{ac618} follow from the identification $\Psi_+(x)=\Psi_{\mathfrak{S} ,\Theta_\infty}(x)$, cf.\ equation~\eqref{hyperg} of Lemma~\ref{lemmahyp}; in particular, \eqref{invhyp} is consistent with the inversion formula~\eqref{invps}. The semi-degenerate spectral type of the Fuchsian singularity at $z=1$ was only necessary to obtain explicit expression of the auxiliary 3-point solution $\Psi_+(z)$. The rest of the argument remains valid even if the singularity at $z=1$ is generic.
\end{Remark}

\section{Multivariate extension}
This section is devoted to a multivariate generalization of the FST$_N$ system obtained by adding extra semi-degenerate singularities. The matrix~$A(z)$ in~\eqref{lis} is replaced by
\begin{gather}\label{ffs}
A(z)=\sum_{k=0}^{n-2}\frac{A_k}{z-z_k},\qquad z_0=0,\qquad z_{n-2}=1,
\end{gather}
where $A_1,\ldots,A_{n-2}$ have rank 1. We set $z_{n-1}=\infty$, $A_{n-1}=-\sum\limits_{k=0}^{n-2}A_k$ and assume radial ordering: $0<|z_1|<\cdots <|z_{n-3}|<1$. The fundamental solution $\Phi(z)$ has monodromy $M_{k}\in\mathrm{GL}(N,\Cb)$ upon analytic continuation around $z_{k}$ ($k=0,\ldots,n-1$), see Fig.~\ref{figg1} for the $n=4$ case. These monodromy matrices satisfy the cyclic relation $M_0\cdots M_{n-1}=\mathds 1$.
It will be convenient for us to consider the products $M_{0\to k}:=M_0\cdots M_k$ and introduce a notation for their spectrum via $[M_{0\to k}]=[{\rm e}^{2\pi{\rm i}\mathfrak S_k}]$, $k=0,\ldots, n-2$, where the eigenvalues of diagonal matrices $\mathfrak S_k$ are assumed to be pairwise distinct $\operatorname{mod} \Zb$. It may also be assumed that $\operatorname{Tr}\mathfrak S_k=\sum\limits_{j=0}^k \operatorname{Tr}\Theta_{j}$. For notational purposes, it is convenient to identify $\mathfrak S_0=\Theta_0$, $\mathfrak S_{n-2}=-\Theta_{n-1}$.

\begin{Remark}The solution of the multivariate extension of the FST$_N$ system corresponding to reducible monodromy (i.e., Riccati-type solution) was constructed in terms of a multivariate hypergeometric series in~\cite{ManoTsuda,Tsuda2}. Our aim in this section is to study the solutions corresponding to generic semi-degenerate monodromy. One of their potential applications is the theory of Frobenius manifolds, where isomonodromic deformations of the Fuchsian systems with degenerate local monodromies naturally arise, see for example \cite[Section~3]{Dubrovin} and \cite[Section~8]{ChM}.
\end{Remark}

The following straightforward generalization of Lemma~\ref{lem31} provides a parameterization of semi-degenerate monodromy \cite{GIL18}.
\begin{Proposition}\label{propparam}
Let $M_k\in \mathrm{GL}(N,\mathbb{C})$ with $k=0,\ldots,n-1$ be the monodromy matrices of the semi-degenerate Fuchsian system satisfying the above genericity condition. Given their spectra $ [M_k ]=[\exp{2\pi\Theta_k}]$, they can be parameterized uniquely (up to an overall conjugation) by means of diagonal matrices $\{\mathfrak{S}_k,D_k\}$ with ${k=1,\ldots,n-3}$, where $D_k=\operatorname{diag}\bigl\{{\rm e}^{\beta^{(1)}_k},\ldots,{\rm e}^{\beta^{(N)}_k}\bigr\}$, \mbox{$\beta^{(l)}_k\in \mathbb{C}$}. The parametrization of $M_k=({M_{0\to {k-1}}})^{-1} M_{0\to k}$ follows from
\begin{gather*}
M_{0\to k}=S_k {\rm e}^{2\pi{\rm i}\mathfrak S_k}S_k^{-1}, \qquad k=0,\ldots, n-2,\\ 
S_{k}^{-1} = D_{k} W_{k+1} D_{k+1} \cdots D_{n-3} W_{n-2} D_{n-2} ,
\end{gather*}
where
\begin{gather*} 
(W_m)_{kl}=\frac{ \prod\limits_{s\ne k} \sin\pi\big(\sigma_{m-1}^{(s)}-\sigma_{m}^{(l)}\big)} { \prod\limits_{s\ne l} \sin\pi\big(\sigma_{m}^{(l)}-\sigma_{m}^{(s)}\big)} .
\end{gather*}
\end{Proposition}

For $n>4$, the circle $\mathcal C$ of the 4-point system is replaced with $n-3$ non-intersecting simple closed curves decomposing $\Cb\Pb^1\ \backslash \{z_0,\ldots, z_{n-1} \}$ into $n-2$ spheres with 3 punctures. A general method to construct Fredholm determinant representation of the isomonodromic tau function in the multi-curve setup is outlined in~\cite{CGL}. Among different topologically inequivalent systems of cutting curves, we are going to use the simplest one, given by a set of concentric cirles
\begin{gather*}
\mathcal C_k=\big\{z\in\Cb\colon |z|=R_k,\, |z_k|<R_k<|z_{k+1}|\big\},\qquad k=1,\ldots, n-3,
\end{gather*}
and corresponding to a linear pants decomposition, cf.~\cite[Section~3.3]{CGL}. Note that the diagonal matrices $\{\mathfrak{S}_k,D_k\}$ may be thought of as associated with the circle $\mathcal C_k$. The analog of~\eqref{topodec} is given by
\begin{gather}
\tau_{\mathrm{JMU}}\left(
\vcenter{\hbox{
\includegraphics[height=8.4ex]{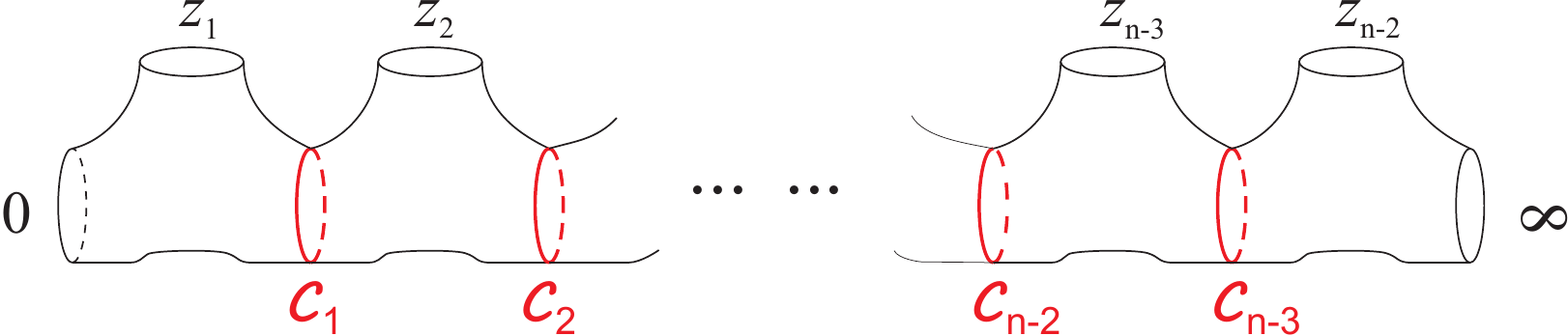}}}\right)
\nonumber\\
\qquad{} =\prod_{k=1}^{n-2}\tau_{\mathrm{JMU}}\left(
\vcenter{\hbox{
\includegraphics[height=8.4ex]{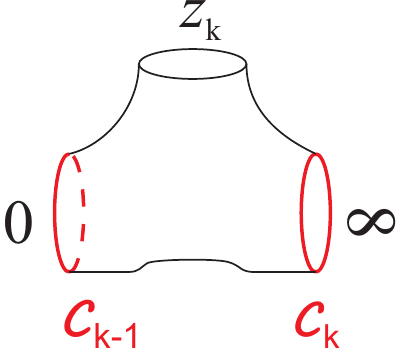}}}\right)
 \det \left(\begin{matrix}
\mathds 1 & \mathsf A_{+-}\\
\mathsf A_{-+} & \mathds 1
\end{matrix}\right),
\label{topodec2}
\end{gather}
where
\begin{gather}\label{Apmblock}
\mathsf A_{+-}=
\left( \begin{array}{@{}cc:cc:cc:cc@{}}
0 & \cdot & \cdot & \cdot & \cdot & 0 & \mathsf a_{\mathcal C_{n-3},+;\mathcal C_{n-3},-} \\
\hdashline
\cdot & \cdot & 0 & 0 & & \iddots & 0\\
\cdot & \cdot & 0 & 0 & \iddots & & \cdot
\\ \hdashline
0 & 0 & \mathsf a_{\mathcal C_4,+;\mathcal C_3,-} &
\mathsf a_{\mathcal C_4,+;\mathcal C_4,-} & 0 & 0 & \cdot \\
0 & 0 & \mathsf a_{\mathcal C_3,+;\mathcal C_3,-} &
\mathsf a_{\mathcal C_3,+;\mathcal C_4,-} & 0 & 0 & \cdot\\
\hdashline
\mathsf a_{\mathcal C_2,+;\mathcal C_1,-} &
\mathsf a_{\mathcal C_2,+;\mathcal C_2,-} & 0 & 0 & \cdot & \cdot & \cdot \\
\mathsf a_{\mathcal C_1,+;\mathcal C_1,-} &
\mathsf a_{\mathcal C_1,+;\mathcal C_2,-} & 0 & 0 & \cdot & \cdot & 0
\end{array} \right), \\
\mathsf A_{-+}=
\left( \begin{array}{@{}cc:cc:cc:cc@{}}
0 & \cdot & \cdot & \cdot & 0 & 0 & \mathsf a_{\mathcal C_{1},-;\mathcal C_{1},+} \\
\hdashline
\cdot & \cdot & 0 & 0 & \mathsf a_{\mathcal C_{2},-;\mathcal C_{3},+} & \mathsf a_{\mathcal C_{2},-;\mathcal C_{2},+} & 0\\
\cdot & \cdot & 0 & 0 & \mathsf a_{\mathcal C_{3},-;\mathcal C_{3},+} & \mathsf a_{\mathcal C_{3},-;\mathcal C_{2},+} & 0
\\ \hdashline
\cdot & \cdot & &
\iddots & 0 & 0 & \cdot \\
0 & 0 & \iddots &
& 0 & 0 & \cdot\\
\hdashline
\mathsf a_{\mathcal C_{n-4},-;\mathcal C_{n-3},+} &
\mathsf a_{\mathcal C_{n-4},-;\mathcal C_{n-4},+} & 0 & \cdot & \cdot & \cdot & \cdot \\
\mathsf a_{\mathcal C_{n-3},-;\mathcal C_{n-3},+} &
\mathsf a_{\mathcal C_{n-3},-;\mathcal C_{n-4},+} & 0 & \cdot & \cdot & \cdot & 0
\end{array} \right).\nonumber
\end{gather}

A few explanations are in order:
\begin{itemize}\itemsep=0pt
\item The tau functions of auxiliary 3-point Fuchsian systems have elementary expressions in terms of monodromy of the initial $n$-point problem:
\begin{gather}
\tau_{\mathrm{JMU}}\left(
\vcenter{\hbox{
\includegraphics[height=8.4ex]{FSTFigmany3}}}\right)= z_k^{\frac12\operatorname{Tr}\big(\mathfrak S_k^2-\mathfrak S_{k-1}^2-\Theta_k^2\big)}.
\end{gather}
Recall that $\mathfrak S_k$ and $\Theta_k$ are diagonal matrices such that $[M_k]=[{\rm e}^{2\pi{\rm i} \Theta_k}]$ and $[M_0\cdots M_k]=\big[{\rm e}^{2\pi{\rm i} \mathfrak S_k}\big]$, where $M_k$ denotes the anti-clockwise monodromy around $z_k$.
\item The structure of the operators $\mathsf A_{\pm\mp}$ is determined by the choice of the pants decomposition and its 2-coloring. We fix them in the same way as in \cite[Section~3.3]{CGL}. The pants with boundary components $\mathcal C_{k-1}$ and $\mathcal C_{k}$ will be denoted by $\mathcal T^{[k]}$. The odd- and even-numbered pairs of pants will have color~``$+$'' and~``$-$''. The orientation of the circles is fixed accordingly: $\mathcal C_{2k-1}$ and $\mathcal C_{2k}$ are oriented anti-clockwise and clockwise, respectively.
\item The rows and columns of $\mathsf A_{\pm\mp}$ are labeled by the curves $\mathcal C_1,\ldots,\mathcal C_{n-3}$. The element of $\mathsf A_{\pm\mp}$ corresponding to a pair $\mathcal C_k$, $\mathcal C_l$ is non-zero only if the two curves are boundary components of the same pair of pants of color ``$\mp$''. The formulas~\eqref{Apmblock} correspond thus to even~$n$. For odd~$n$ the block in the upper-right corner of $\mathsf A_{+-}$ would be $2\times 2$, whereas the block in the bottom-left corner of $\mathsf A_{-+}$ would be $1\times 1$.
\end{itemize}
Denote $H_{\mathcal C}=L^2\big(\mathcal C,\Cb^N\big)$ and let $H_{\mathcal C,\pm}$ be the space of boundary values of functions that continue analytically to the interior/exterior of $\mathcal C$ with respect to its orientation. The determinant in~\eqref{topodec2} is computed on $H=\bigoplus_{k=1}^{n-3} H_{\mathcal C_k}$. The operators $\mathsf a_{\mathcal C,\pm;\mathcal C',\mp}$ act from $H_{\mathcal C',\mp}$ to $H_{\mathcal C,\pm}$ as follows:
\begin{gather}\label{amm1}
( \mathsf a_{\mathcal C,\pm;\mathcal C',\mp} g)(z)=\frac1{2\pi{\rm i}}\oint_{\mathcal C'}\mathsf a_{\mathcal C,\pm;\mathcal C',\mp}(z,z') g(z') {\rm d}z',\qquad z\in\mathcal C,
\end{gather}
with the integral kernel
\begin{gather}\label{amm2}
\mathsf a_{\mathcal C,\pm;\mathcal C',\mp}(z,z')= \pm \frac{\Psi_{\mathcal C,\pm}(z) {\Psi_{\mathcal C',\pm}(z')}^{-1}-\mathds 1 \delta_{\mathcal C,\mathcal C'}}{z-z'}.
\end{gather}
The functions $\Psi_{\mathcal C,\pm}(z)$, $\Psi_{\mathcal C',\pm}(z)$ are determined by the fundamental solutions of the auxiliary 3-point Fuchsian systems associated to appropriate pairs of pants, see for instance \cite[equations~(3.9)]{CGL}. Note that each $2\times2$ block of $\mathsf A_{\pm\mp}$ involves only one such 3-point solution; the blocks could thus be labeled by~$\mathcal T^{[k]}$.

Denote by $\Phi^{[k]}(z)$ ($k=1,\ldots,n-2$) the solution of 3-point Fuchsian system associated with the pants $\mathcal T^{[k]}$ which has regular singularities at $0$, $z_k$ and $\infty$ characterized by monodromies $M_{0\to k-1}$, $M_{k}$ and $M_{0\to k}^{-1}$. For consistency with Proposition~\ref{propparam}, the local behavior of this solution near the singular points is required to be given by
\begin{gather*}
\Phi^{[k]}(z) = \begin{cases}
S_{k-1} (-z)^{\mathfrak S_{k-1}} G_0^{[k]}(z), & z\to0, \\
 S_{k} (-z)^{\mathfrak S_{k}} G_{\infty}^{[k]}(z), & z\to \infty,
\end{cases}
\end{gather*}
where $G^{[k]}_{0}(z)$, $G^{[k]}_{\infty}(z)$ are holomorphic and invertible in the respective neighborhoods of~$0$ and~$\infty$. Using Lemma~\ref{lemmahyp}, we may fix
\begin{gather*}
\Phi^{[k]}(z) = S_{k-1} \Phi^{(0)}\big(\tfrac{z}{z_k} \big)=S_{k} \Phi^{(\infty)}\big(\tfrac{z}{z_k}\big),
\end{gather*}
where the parameters of the lemma are identified as $\Theta_0\mapsto \mathfrak{S}_{k-1}$, $\Theta_\infty\mapsto -\mathfrak{S}_{k}$. Furthermore, using the relation between $\Phi^{(0)}\big(\frac{z}{z_k}\big)$ and $\Phi^{(\infty)}\big(\frac{z}{z_k}\big)$ from the same lemma together with the parametrization of Proposition~\ref{propparam}, we obtain
\begin{gather*}
S_{k-1}^{-1}S_k = D_{k-1}W_k=\bar R^{[k]} S_{\mathfrak{S}_{k-1},-\mathfrak{S}_{k}} \big(R^{[k]}\big)^{-1}.
\end{gather*}
The latter relation connects the diagonal matrices $R^{[k]}$ and $\bar R^{[k]}$ which appear in Lemma~\ref{lemmahyp} to the parameters of~$D_k$ in Proposition~\ref{propparam}:
\begin{gather}\label{RbRres}
\bar r^{[k+1]}_m={\rm e}^{\beta_k^{(m)}}\frac{\prod\limits_l \Gamma\big(\sigma^{(m)}_k -\sigma^{(l)}_{k+1} \big)}{ \prod\limits_{l\ne m}\Gamma\big(1+\sigma^{(m)}_k-\sigma^{(l)}_k\big)},\qquad
r^{[k]}_j=\frac {\prod\limits_l \Gamma\big(\sigma^{(l)}_{k-1} - \sigma^{(j)}_k \big)} {\prod\limits_{l \ne j}\Gamma\big(1- \sigma^{(j)}_k +\sigma^{(l)}_k \big)} .
\end{gather}
The final formulas for $\Psi_{\mathcal C,\pm}(z)$ are then given by
\begin{gather*}
\Psi_{\mathcal {C}_{k,(-)^k}}(z)=R^{[k]} z_k^{-\mathfrak{S}_k}\Psi_{-\mathfrak{S}_k,\mathfrak{S}_{k-1}}\big(\tfrac{z_k}{z} \big) \big(R^{[k]}\big)^{-1},\\
\Psi_{\mathcal {C}_{k-1,(-)^{k}}}(z)= \bar R^{[k]} z_{k}^{-\mathfrak{S}_{k-1}}\Psi_{\mathfrak{S}_{k-1},-\mathfrak{S}_{k}}\big(\tfrac{z}{z_{k}}\big) G_{-\mathfrak{S}_{k},\mathfrak{S}_{k-1}}\big(R^{[k]}\big)^{-1},
\end{gather*}
where $\Psi_{\mathfrak S,\mathfrak S'}(z)$, $G_{\mathfrak S,\mathfrak S'}$, $R^{[k]}$ and $\bar{R}^{[k]}$ are defined by (\ref{hyperg}), (\ref{exprG0}) and (\ref{RbRres}). In combination with \eqref{Apmblock}--\eqref{amm2}, these formulas make the Fredholm determinant representation~\eqref{topodec2} completely explicit for the tau functions of semi-degenerate Fuchsian systems.

The series representation for $\tau(\boldsymbol z)$ is obtained similarly to the $4$-point case: it suffices to rewrite the operators $\mathsf A_{\pm\mp}$ in the Fourier basis and expand the Fredholm determinant into a~sum of the principal minors. The structure of \eqref{Apmblock} implies that these minors factorize into products of smaller ones associated to different $\mathcal T^{[k]}$, which is a~consequence of the linear pants decomposition that we are using. The minors are labeled by $n-3$ $N$-tuples of Maya diagrams $(-\boldsymbol{\mathsf q}_k,\boldsymbol{\mathsf p}_k)$ associated to each circle $\mathcal C_k$ ($k=1,\ldots,n-3$). Furthermore, they are non-zero only if each of the $N$-tuples satisfies the condition of global balance $|\boldsymbol{\mathsf p}_k|=|\boldsymbol{\mathsf q}_k|$. Therefore, the determinant in~\eqref{topodec2} can be represented as
\begin{gather}
\det \left(\begin{matrix}
\mathds 1 & \mathsf A_{+-}\\
\mathsf A_{-+} & \mathds 1
\end{matrix}\right)=\big[Z\big(\mathcal T^{[n-2]}\big)\cdots Z\big(\mathcal T^{[2]}\big) Z\big(\mathcal T^{[1]}\big)\big]_{\varnothing,\varnothing
}^{\varnothing,\varnothing}\nonumber\\
\qquad{} =\sum_{\substack{\boldsymbol{\mathsf p}_1,\boldsymbol{\mathsf q}_1,\ldots, \boldsymbol{\mathsf p}_{n-3},\boldsymbol{\mathsf q}_{n-3}\\ |\boldsymbol{\mathsf p}_1|=|\boldsymbol{\mathsf q}_1|,\ldots, |\boldsymbol{\mathsf p}_{n-3}|=|\boldsymbol{\mathsf q}_{n-3}|}}
Z_{\quad\varnothing,\varnothing}^{\boldsymbol{\mathsf p}_{n-3},\boldsymbol{\mathsf q}_{n-3}}\big(\mathcal T^{[n-2]}\big)
\cdots
Z^{\boldsymbol{\mathsf p}_{1},\boldsymbol{\mathsf q}_{1}}_{\boldsymbol{\mathsf p}_{2},\boldsymbol{\mathsf q}_{2}}\big(\mathcal T^{[2]}\big)
Z^{\varnothing,\varnothing}_{\boldsymbol{\mathsf p}_{1},\boldsymbol{\mathsf q}_{1}}\big(\mathcal T^{[1]}\big),\label{detAZZZ}
\end{gather}
where, e.g., for pants $\mathcal T^{[l]}$ with $l=2k+1$ we may set
\begin{gather}\label{eld}
Z^{\boldsymbol{\mathsf p}_{l-1},\boldsymbol{\mathsf q}_{l-1}}_{\boldsymbol{\mathsf p}_{l},\boldsymbol{\mathsf q}_{l}}\big(\mathcal T^{[l]}\big) =\det
\begin{blockarray}{ccl}
{\text{\footnotesize$\boldsymbol{\mathsf p}_{l}$}} & {\text{\footnotesize$-\boldsymbol{\mathsf q}_{l-1}$}} \\
\begin{block}{(cc)l}
\mathsf a_{\mathcal C_{l-1},-;\mathcal C_l,+} &
\mathsf a_{\mathcal C_{l-1},-;\mathcal C_{l-1},+} & {\text{\footnotesize$\boldsymbol{\mathsf p}_{l-1}$}} \\
\mathsf a_{\mathcal C_{l},-;\mathcal C_l,+} &
\mathsf a_{\mathcal C_{l},-;\mathcal C_{l-1},+} & {\text{\footnotesize$-\boldsymbol{\mathsf q}_{l}$}} \\
\end{block}
\end{blockarray}.
\end{gather}
The above notation means that the determinant is composed of matrix elements from columns $\boldsymbol{\mathsf p}_{l}$, $-\boldsymbol{\mathsf q}_{l-1}$ and rows $\boldsymbol{\mathsf p}_{l-1}$, $-\boldsymbol{\mathsf q}_{l}$ of $\mathsf a$'s written in the Fourier basis:
\begin{subequations}
\begin{alignat}{3}
\label{maA}
& \mathsf a_{\mathcal C_{l},-;\mathcal C_{l},+}(z,z')=\sum_{p,q\in\mathbb Z'_+}
\mathsf a_{\mathcal C_{l},-q;\mathcal C_{l},p}z^{-q-\frac12}z'^{-p-\frac12},\qquad && z,z'\in\mathcal C_l,&\\
\label{maB} &\mathsf a_{\mathcal C_{l},-;\mathcal C_{l-1},+}(z,z')=\sum_{p,q\in\mathbb Z'_+}
\mathsf a_{\mathcal C_{l},-q;\mathcal C_{l-1},-p}z^{-q-\frac12}z'^{p-\frac12},\qquad && z\in\mathcal C_{l},\quad z'\in\mathcal C_{l-1},&\\
\label{maC}&\mathsf a_{\mathcal C_{l-1},-;\mathcal C_{l},+}(z,z')=\sum_{p,q\in\mathbb Z'_+}
\mathsf a_{\mathcal C_{l-1},p;\mathcal C_{l},q}z^{p-\frac12}z'^{-q-\frac12},\qquad &&z\in\mathcal C_{l-1},\quad z'\in\mathcal C_l, &\\
\label{maD}&\mathsf a_{\mathcal C_{l-1},-;\mathcal C_{l-1},+}(z,z')=\sum_{p,q\in\mathbb Z'_+}
\mathsf a_{\mathcal C_{l-1},p;\mathcal C_{l-1},-q}z^{p-\frac12}z'^{q-\frac12},\qquad &&z,z'\in\mathcal C_{l-1}.&
\end{alignat}
\end{subequations}
Equivalents of the matrix elements appearing in the decompositions \eqref{maA} and \eqref{maD} were effectively calculated in the 4-point case:
\begin{gather}
\mathsf a_{\mathcal C_{l},-q,j;\mathcal C_{l},p,k}= -\frac{r^{[l]}_j}{r^{[l]}_k}\frac{\varphi_{-q}^{(j)}(\mathfrak S_{l-1},-\mathfrak S_l)
\bar\varphi_{p}^{(k)}(\mathfrak S_{l-1},-\mathfrak S_l) z_l^{q-\sigma_l^{(j)}+p+\sigma_l^{(k)}}}{q-\sigma_l^{(j)}+p+\sigma_l^{(k)}},\nonumber\\
\mathsf a_{\mathcal C_{l-1},p,j;\mathcal C_{l-1},-q,k}= -\frac{\bar r^{[l]}_j}{\bar r^{[l]}_k}
\frac{\varphi_{-p}^{(j)}(-\mathfrak S_l,\mathfrak S_{l-1}) \bar\varphi_{q}^{(k)}(-\mathfrak S_l,\mathfrak S_{l-1})
z_l^{-p-\sigma_{l-1}^{(j)}-q+\sigma_{l-1}^{(k)}}}{p+\sigma_{l-1}^{(j)}+q-\sigma_{l-1}^{(k)}} .\label{cael1}
\end{gather}
Here $\varphi_{-q}^{(j)}(\mathfrak S,-\mathfrak S')$, $\bar\varphi_{p}^{(k)}(\mathfrak S,-\mathfrak S')$ are defined by \eqref{lemapq2}. The matrix elements in~\eqref{maB} and \eqref{maC} may be obtained by a similar procedure:
\begin{gather}
\mathsf a_{\mathcal C_{l-1},p,j;\mathcal C_{l},q,k}= \frac{\bar r^{[l]}_j}{ r^{[l]}_k} \frac{\varphi_{-p}^{(j)}(-\mathfrak S_l,\mathfrak S_{l-1})
\bar\varphi_{q}^{(k)}(\mathfrak S_{l-1},-\mathfrak S_l) z_l^{-p-\sigma_{l-1}^{(j)}+q+\sigma_{l}^{(k)}}}{p+\sigma_{l-1}^{(j)}-q-\sigma_{l}^{(k)}}, \nonumber\\
\mathsf a_{\mathcal C_{l},-q,j;\mathcal C_{l-1},-p,k}= \frac{ r^{[l]}_j}{\bar r^{[l]}_k}\frac{\varphi_{-q}^{(j)}(\mathfrak S_{l-1},-\mathfrak S_l) \bar\varphi_{p}^{(k)} (-\mathfrak S_l,\mathfrak S_{l-1}) z_l^{q-\sigma_{l}^{(j)}-p+\sigma_{l-1}^{(k)}}}{q-\sigma_{l}^{(j)}-p+\sigma_{l-1}^{(k)}} .\label{cael2}
\end{gather}
Thus, up to an overall sign, the contributions $ Z^{\boldsymbol{\mathsf p}_{l-1},\boldsymbol{\mathsf q}_{l-1}}_{\boldsymbol{\mathsf p}_{l},\boldsymbol{\mathsf q}_{l}}\big(\mathcal T^{[l]}\big) $ of individual pairs of pants in~\eqref{eld} are given by Cauchy-type determinants:
\begin{gather}
Z^{\boldsymbol{\mathsf p}_{l-1},\boldsymbol{\mathsf q}_{l-1}}_{\boldsymbol{\mathsf p}_{l},\boldsymbol{\mathsf q}_{l}}\big(\mathcal T^{[l]}\big) =
\pm z_l^{\boldsymbol{\sigma}_l\cdot\boldsymbol{w}_l-\boldsymbol{\sigma}_{l-1}\cdot\boldsymbol{w}_{l-1}+\sum\limits_{q\in\boldsymbol{\mathsf q}_{l}} q+\sum\limits_{p\in\boldsymbol{\mathsf p}_{l}} p-\sum\limits_{q\in\boldsymbol{\mathsf q}_{l-1}}q-\sum\limits_{p\in\boldsymbol{\mathsf p}_{l-1}}p }
\nonumber\\
\hphantom{Z^{\boldsymbol{\mathsf p}_{l-1},\boldsymbol{\mathsf q}_{l-1}}_{\boldsymbol{\mathsf p}_{l},\boldsymbol{\mathsf q}_{l}}\big(\mathcal T^{[l]}\big)}{} \times
\prod_{k=1}^N\big(\bar r_k^{[l]}\big)^{w^{(k)}_{l-1}}\big( r_k^{[l]}\big)^{-w^{(k)}_{l}} Z^{\boldsymbol{\mathsf p}_{l-1},\boldsymbol{\mathsf q}_{l-1}}_{\boldsymbol{\mathsf p}_{l},\boldsymbol{\mathsf q}_{l}} (\mathfrak S_{l-1},\mathfrak S_l ),\label{detzpq}
\end{gather}
where
\begin{gather}
Z^{\boldsymbol{\mathsf p}',\boldsymbol{\mathsf q}'}_{\boldsymbol{\mathsf p},\boldsymbol{\mathsf q}}(\mathfrak S', \mathfrak S)
= \prod_{(p,j)\in\boldsymbol{\mathsf p}'} \varphi_{-p}^{(j)}(-\mathfrak S,\mathfrak S')
\prod_{(q,k)\in\boldsymbol{\mathsf q}'}\bar\varphi_{q}^{(k)}(-\mathfrak S,\mathfrak S')\nonumber\\
\hphantom{Z^{\boldsymbol{\mathsf p}',\boldsymbol{\mathsf q}'}_{\boldsymbol{\mathsf p},\boldsymbol{\mathsf q}}(\mathfrak S', \mathfrak S)=}{}
\times \prod_{(q,j)\in\boldsymbol{\mathsf q}} \varphi_{-q}^{(j)}(\mathfrak S',-\mathfrak S)
\prod_{(p,k)\in\boldsymbol{\mathsf p}} \bar\varphi_{p}^{(k)}(\mathfrak S',-\mathfrak S) \nonumber\\
\hphantom{Z^{\boldsymbol{\mathsf p}',\boldsymbol{\mathsf q}'}_{\boldsymbol{\mathsf p},\boldsymbol{\mathsf q}}(\mathfrak S', \mathfrak S)=}{} \times
\det
\begin{blockarray}{ccl}
{\text{\footnotesize$\boldsymbol{\mathsf p}$}} & {\text{\footnotesize$-\boldsymbol{\mathsf q}'$}} \vspace{0.2cm} \\
\begin{block}{(cc)l}
\displaystyle \frac{1}{p'+{\sigma'}^{(j)}-p-\sigma^{(k)}} &
\displaystyle \frac1{p'+{\sigma'}^{(j)}+q'-{\sigma'}^{(k)}} & {\text{\footnotesize$\boldsymbol{\mathsf p}'$}} \vspace{0.1cm} \\
\displaystyle \frac1{-q+\sigma^{(j)}-p-\sigma^{(k)}} &
\displaystyle \frac{1}{-q+\sigma^{(j)}+q'-{\sigma}^{(k)}} & {\text{\footnotesize$-\boldsymbol{\mathsf q}$}} \\
\end{block}
\end{blockarray},\label{detzpqnorm}
\end{gather}
and $\boldsymbol{w}_l\in\mathfrak{Q}_{N-1}$ denotes the charge vector of the $N$-tuple of Maya diagrams assigned to the circle $\mathcal C_l$. Substituting the expressions \eqref{detzpq} into the principal minor expansion, we obtain a~multivariate analog of the series representation \eqref{sertaufst} of the FST$_N$ tau function. The result is expressed in terms of $\boldsymbol{\sigma}_1,\ldots, \boldsymbol{\sigma}_{n-3}$ and $\boldsymbol{\beta}_1,\ldots, \boldsymbol{\beta}_{n-3}$, the latter parameters entering only via~\eqref{RbRres}.

\appendix

\section{AGT-W relation from isomonodromy}\label{appendixA}
We start with a brief review of semi-degenerate vertex operators of $W_N$-algebras and the corresponding conformal blocks (CBs), trying to highlight the general structures without entering into the details. The $W_N$-algebras are infinite-dimensional associative algebras generalizing Virasoro algebra which is $W_2$-algebra. The Verma modules of $W_N=W(\mathfrak{sl}_N)$ are labeled by the weights $\bs \sigma=\big(\sigma^{(1)},\ldots,\sigma^{(N)}\big)\in \mathbb{C}^N$, $\sum\limits_{\alpha=1}^N \sigma^{(\alpha)}=0$, of the Lie algebra~$\mathfrak{sl}_N$. These Verma modules are irreducible for generic $\bs\sigma$. We denote by $\mathcal{V}_{\bs \sigma}$ the irreducible Verma modules or irreducible quotients with the same highest weight vectors $ |\bs\sigma \ra$ of reducible Verma modules. Important examples of irreducible quotients are semi-degenerate $W_N$-modules $\mathcal{V}_{\bs \sigma}$ with $\bs\sigma=\Lambda \bs h_1$, where $\Lambda\in \mathbb{C}$ and $\bs h_1=\big(\frac{N-1}{N},-\frac1N,\ldots,-\frac1N\big)$ is the 1st fundamental weight of~$\mathfrak{sl}_N$.

In order to define CBs of the $W_N$-algebra on $\mathbb{CP}^1$, one has to fix~$n$ points $\bs z:=\{z_0\equiv0,z_1,\ldots,$ $z_{n-2},z_{n-1}\equiv\infty\}$ and attach a $W_N$-module $\mathcal{V}_{\bs \theta_k}$ to each of them. The $n$-point CBs are linear forms on ${\mathcal V}_{\bs \theta_0}\otimes \cdots \otimes {\mathcal V}_{\bs \theta_{n-1}}$ satisfying the so-called conformal Ward identities. In general, the space of such CBs is infinite-dimensional. For its description, it is useful to decompose $\mathbb{CP}^1\backslash{\bs z}$ into pairs of pants as above. With each pair of pants $\mathcal T^{[k]}$ ($k=1,\ldots,n-2$) is associated a space of $3$-point CBs if we fix $W_N$-modules $\mathcal{V}_{\bs \sigma_l}$ on cutting circles~$\mathcal{C}_l$ ($l=1,\ldots,n-3$). Such $3$-point CBs are linear forms on $\mathcal{V}_{\bs \sigma_k}\otimes \mathcal{V}_{\bs \theta_k} \otimes \mathcal{V}_{\bs \sigma_{k-1}}$, where $\bs\sigma_0=\bs\theta_0$ and $\bs\sigma_{n-2}=-\bs\theta_{n-1}$. For generic~${\bs \sigma_k}$, ${\bs \theta_k}$, ${\bs \sigma_{k-1}}$, this space of $3$-point CBs is also infinite-dimensional. However, in certain special cases, e.g., for semi-degenerate~$\mathcal{V}_{\bs \theta_k}$ (with $\bs \theta_k=\Lambda_k \bs h_1$), it becomes one-dimensional, similarly to the space of generic $3$-point Virasoro CBs.

Using the Shapovalov non-degenerate bilinear form on irreducible modules $\mathcal{V}_{\bs \sigma}$, one can rewrite linear forms on $\mathcal{V}_{\bs \sigma_k}\otimes \left|{\bs \theta_k}\right\ra \otimes \mathcal{V}_{\bs \sigma_{k-1}}$, where $|{\bs \theta_k}\ra $ is highest weight vector of $\mathcal{V}_{\bs \theta_k}$, as linear operators $V_{\bs \theta_k} (z_k)\colon \mathcal{V}_{\bs \sigma_{k-1}} \to \mathcal{V}_{\bs \sigma_{k}}$, which are called vertex operators. In the case of semi-degenerate $\bs \theta_k=\Lambda_k \bs h_1$, the space of such operators is one-dimensional, and they are uniquely fixed by their matrix elements between the highest weight vectors. E.g., for the central charge $c=N-1$, they read
\begin{gather}\label{VOnorm}
\big\la \bs \sigma_{k}\,|\, V_{\Lambda_k}(z_k)\,|\, \bs \sigma_{k-1}\big\ra=\mathcal{N}(\bs \sigma_{k}, \Lambda_k \bs h_1, \bs \sigma_{k-1})
z_k^{\frac12\big(\bs \sigma_{k}^2-\Lambda_k^2 \bs h_1^2-\bs \sigma_{k-1}^2\big)},
\end{gather}
where we used a shorthand notation $V_{\Lambda_k}(z_k)$ for the corresponding semi-degenerate vertex ope\-rator.

In \cite{AGT}, a relation between instanton partition functions and Virasoro CBs was conjectured. Later in~\cite{Wyl}, this proposal was extended to semi-degenerate CBs of $W_N$-algebras. The conjectured relation is as follows~\cite{FL3}: there exist bases of $W_N$-modules $\mathcal{V}_{\bs \sigma}$ labeled by $N$-tuples of Young diagrams $\bs Y=\big(Y^{(1)},\ldots,Y^{(N)}\big)\in\mathbb Y^N$ such that the matrix elements of semi-degenerate vertex operators are given by (for simplicity, we again restrict ourselves to $c=N-1$ which is the only case relevant to our purposes)
\begin{gather}\label{GTF}
F_{\bs Y,\bs Y'}(\bs\sigma,\Lambda,\bs\sigma'):=
\frac{ \la \bs\sigma, \bs Y\,|\,V_\Lambda(1)\,|\,\bs\sigma', \bs Y' \ra}
{ \la \bs\sigma\,|\,V_\Lambda(1)\,|\,\bs\sigma' \ra}=\prod_{\alpha,\beta=1}^N Z_{\mathsf{bif}}
\big({\sigma'}^{(\alpha)}- \sigma^{(\beta)} -\tfrac{\Lambda}{N}\,\big|\, {Y'}^{(\alpha)},Y^{(\beta)}\big),\!\!\!
\end{gather}
where
\begin{gather}\label{Zbif}
Z_{\mathsf{bif}}(\nu|Y',Y):=\prod_{\square\in Y'}\bigl(\nu+1+a_{Y'}(\square)+l_{Y}(\square)\bigr)
\prod_{\square\in Y}\bigl(\nu-1-a_{Y}(\square)-l_{Y'}(\square)\bigr),
\end{gather}
and $a_{Y}(\square)$, $l_{Y}(\square)$ denote the arm-length and leg-length of the box $\square$ with respect to $Y$. The basis $\{|\bs\sigma, \bs Y\ra\}_{\bs Y\in\mathbb Y^N}$ is orthogonal with respect to the Shapovalov form:
\begin{gather*}
 \la \bs\sigma, \bs Y\,|\,\bs\sigma', \bs Y' \ra= \delta_{\bs\sigma,\bs\sigma'} \delta_{\bs Y,\bs Y'}F_{\bs Y,\bs Y}(\bs\sigma,0,\bs\sigma),
\end{gather*}
corresponding to the matrix elements of the identity vertex operator $V_{\bs 0}(1)$. Therefore, the projection $\mathcal{P}_{\bs \sigma}$ to the $W_N$-module $\mathcal{V}_{\bs \sigma}$ is
\begin{gather*}
\mathcal{P}_{\bs \sigma}=\sum_{\bs Y}\frac{|\bs\sigma, \bs Y\ra \la \bs\sigma, \bs Y|}{\la \bs\sigma, \bs Y\,|\,\bs\sigma, \bs Y\ra}.
\end{gather*}
Combining these ingredients together, one comes to combinatorial expressions for semi-degene\-ra\-te CBs. For instance, the 4-point CBs are given by
\begin{gather*}
\frac{\la -\bs\theta_\infty\,|\,V_{\Lambda_3}(1)\mathcal{P}_{\bs\sigma} V_{\Lambda_2}(z)\,|\,\bs\theta_0 \ra}
{\la -\bs\theta_\infty\,|\,V_{\Lambda_3}(1) \,|\,{\bs\sigma} \ra \la{\bs\sigma}\,|\,V_{\Lambda_2}(1)|\bs\theta_0 \ra}
 = (1-z)^{-\frac{\Lambda_3 \Lambda_2}{N}} z^{\frac12\big(\bs\sigma^2-\Lambda_2^2 \bs h_1^2-\bs\theta_0^2\big)}
\nonumber \\
\hphantom{\frac{\la -\bs\theta_\infty\,|\,V_{\Lambda_3}(1)\mathcal{P}_{\bs\sigma} V_{\Lambda_2}(z)\,|\,\bs\theta_0 \ra}
{\la -\bs\theta_\infty\,|\,V_{\Lambda_3}(1) \,|\,{\bs\sigma} \ra \la{\bs\sigma}\,|\,V_{\Lambda_2}(1)|\bs\theta_0 \ra}=}{} \times\sum_{\bs Y\in\mathbb Y^N}
\frac {F_{\bs\varnothing, \bs Y}(-\bs\theta_\infty,\Lambda_3,\bs\sigma) F_{ \bs Y, \bs \varnothing}(\bs\sigma,\Lambda_2,\bs \theta_0)}{F_{\bs Y, \bs Y}(\bs\sigma,0,\bs\sigma)} z^{|\bs Y|}.
\end{gather*}

As was shown in our previous work \cite{GIL18} (generalizing the connection of Painlev\'e VI and Garnier system to Liouville CFT \cite{BSh, GIL12,ILTe}), the tau function $\tau(\boldsymbol{z})$ of the $n$-point semi-degenerate Fuchsian system can be represented as a Fourier transform of $n$-point semi-degenerate $W_N$ conformal blocks.
\begin{Theorem} The tau function $\tau(\boldsymbol{z})$ can be expressed in terms of semi-degenerate $W_N$ conformal blocks with central charge $c=N-1$ as
\begin{gather}\label{FFF}
\tau({\bs z})=
\sum_{\boldsymbol{w}_1,\ldots,\boldsymbol{w}_{n-3}\in \mathfrak Q_{N-1}}
{\rm e}^{\boldsymbol{\beta}_1\cdot\boldsymbol{w}_1+\dots+\boldsymbol{\beta}_{n-3}\cdot\boldsymbol{w}_{n-3}}
\mathcal{F} (\bs\sigma_1+\bs w_1, \ldots, \bs\sigma_{n-3}+\bs w_{n-3}), \\
\mathcal{F} (\bs\sigma_1, \ldots, \bs\sigma_{n-3})= \big\la {-}\bs \theta_{n-1}\,|\,
V_{\Lambda_{n-2}}(z_{n-2}) \mathcal{P}_{\boldsymbol{\sigma}_{n-3}}
V_{\Lambda_{n-3}}(z_{n-3})\mathcal{P}_{\boldsymbol{\sigma}_{n-4}} \cdots \mathcal{P}_{\boldsymbol{\sigma}_1}
V_{\Lambda_{1}}(z_{1})\,|\,\boldsymbol{\theta}_0\big\rangle,\nonumber
\end{gather}
where $(\bs\sigma_1,\ldots,\bs\sigma_{n-3})$ and $(\bs\beta_1,\ldots,\bs\beta_{n-3})$ are parameters fixing the initial conditions, directly related to monodromy parameterization in Proposition~{\rm \ref{propparam}}:
\begin{gather*}
\mathfrak{S}_l\to\operatorname{diag} \{\bs\sigma_l \}+\frac{\sum\limits_{j=1}^l\Lambda_j}{N} \mathds 1,\qquad
D_l=\operatorname{diag}\big\{{\rm e}^{\,\beta_l^{(1)}},\ldots,{\rm e}^{\,\beta_l^{(N)}}\big\}.
\end{gather*}
The vertex operators $V_\Lambda(z)$ are normalized by \eqref{VOnorm} with
\begin{gather}\label{Nnormcoef}
\cN(\bs\sigma,\Lambda\bs h_1,\bs\sigma')=
\frac{\prod\limits_{l,j} G\big(1+\sigma'^{(l)}-\sigma^{(j)}-\tfrac{\Lambda}{N}\big)}
{\prod\limits_{k<m}G\big(1+\sigma'^{(k)}-\sigma'^{(m)}\big)G\big(1-\sigma^{(k)}+\sigma^{(m)}\big)},
\end{gather}
where $G(x)$ denotes the Barnes $G$-function.
\end{Theorem}

We would like to compare this expression for $\tau(\boldsymbol{z})$ with the one coming from (\ref{topodec2})~-- namely, with the combinatorial expansion~(\ref{detAZZZ}) in terms of Cauchy determinants (\ref{detzpq}). It will be shown that the two expansions coincide termwise. Besides the check of consistency of the two approaches, this yields as a byproduct a direct proof of Nekrasov/AGT-W formulas for semi-degenerate conformal blocks of the $W_N$ algebra with $c=N-1$. We are going to proceed along the lines of \cite[Appendix~A]{GL16}, which deals with the $N=2$ case but contains several useful results independent of~$N$.

We will need an extension of the formula (\ref{Zbif}) for charged Young diagrams $Y,Y'\in\mathbb Y$ with charges $w,w'\in\mathbb{Z}$:
\begin{gather*}
\tilde Z_{\mathsf{bif}}(\nu\,|\,w',Y';w,Y)= \prod_i(-\nu)_{q_i'+\frac12} \prod_i
(\nu+1)_{q_i-\frac12} \prod_i(-\nu)_{p_i+\frac12} \prod_i(\nu+1)_{p_i'-\frac12} \\
\hphantom{\tilde Z_{\mathsf{bif}}(\nu\,|\,w',Y';w,Y)=}{} \times \frac{\prod\limits_{i,j}(\nu-q_i'-p_j)\prod\limits_{i,j}(\nu+p_i'+q_j)}{
\prod\limits_{i,j}(\nu-q_i'+q_j)\prod\limits_{i,j}(\nu+p_i'-p_j)},
\end{gather*}
where we used the Frobenius coordinates $({\mathsf p},{\mathsf q})$ of charged Young diagrams $(w,Y)$. It follows from Theorems~A.1 and~A.5 of~\cite{GL16} that $ \tilde Z_{\mathsf{bif}}(\nu\,|\,w',Y';w,Y)$ for charged Young diagrams can be reduced to an expression for non-charged ones:
\begin{gather}\label{chY}
\tilde Z_{\mathsf{bif}}(\nu\,|\,w',Y';w,Y)={ \pm }C(\nu\,|\,w',w) Z_{\mathsf{bif}}(\nu+w'-w\,|\,Y',Y).
\end{gather}
Here
\begin{gather*}
C(\nu\,|\,w',w)\equiv C(\nu\,|\,w'-w)=\frac{G(1+\nu+w'-w)}{G(1+\nu)\Gamma(1+\nu)^{w'-w}}.
\end{gather*}
The only property of $G(x)$ relevant for us is the recurrence relation $G(x+1)=\Gamma(x) G(x)$. In order to lighten the formulas, the calculations below will be carried out only up to an overall sign which can be controlled by considering a suitable limit of parameters, e.g., $\nu\to +{\rm i}\infty$ in~\eqref{chY}.

Our starting point is the expression (\ref{detzpqnorm}). Rewrite the Cauchy determinant therein in the factorized form and compare $Z^{\boldsymbol{\mathsf p}',\boldsymbol{\mathsf q}'}_{\boldsymbol{\mathsf p},\boldsymbol{\mathsf q}} (\mathfrak S', \mathfrak S)$ with the following Cauchy-like factorized expression:
\begin{gather*}
{\tilde{Z}}^{\bs Y', \bs w'}_{\bs Y, \bs w} (\mathfrak S', \mathfrak S)=
\prod_{j=1}^N \bigl|\tilde Z_{\mathsf{bif}}\big(0\,|\,w^{(j)},Y^{(j)};w^{(j)},Y^{(j)}\big)\bigr|^{-\frac12}
\bigl|\tilde Z_{\mathsf{bif}}\big(0\,|\,{w'}^{(j)},{Y'}^{(j)};{w'}^{(j)},{Y'}^{(j)}\big)\bigr|^{-\frac12}\\
{}\times
\frac{\prod\limits_{j,k=1}^N \tilde Z_{\mathsf{bif}}\big({\sigma'}^{(j)}-\sigma^{(k)} \,|\,{w'}^{(j)},{Y'}^{(j)};w^{(k)},Y^{(k)}\big)}{
\!\!\prod\limits_{1\le j<k \le N}\!\!\tilde Z_{\mathsf{bif}}\bigl({\sigma'}^{(j)}-{\sigma'}^{(k)}\,|\,{w'}^{(j)},{Y'}^{(j)};{w'}^{(k)},{Y'}^{(k)}\bigr)
{ \tilde Z_{\mathsf{bif}}\big({\sigma}^{(k)}-{\sigma}^{(j)}\,|\,w^{(k)},Y^{(k)};w^{(j)},Y^{(j)}\big)}}.
\end{gather*}
The result of comparison is
\begin{gather*}\label{zpqzyw}
Z^{\boldsymbol{\mathsf p}',\boldsymbol{\mathsf q}'}_{\boldsymbol{\mathsf p},\boldsymbol{\mathsf q}}
(\mathfrak S', \mathfrak S) =
{\tilde{Z}}^{\bs Y', \bs w'}_{\bs Y, \bs w} (\mathfrak S', \mathfrak S) \prod_{(p,j)\in\boldsymbol{\mathsf p}'}
\delta \varphi^{(j)}(-\mathfrak S,\mathfrak S')\\
\hphantom{Z^{\boldsymbol{\mathsf p}',\boldsymbol{\mathsf q}'}_{\boldsymbol{\mathsf p},\boldsymbol{\mathsf q}}
(\mathfrak S', \mathfrak S) =}{}\times \prod_{(q,k)\in\boldsymbol{\mathsf q}'}
\delta \bar\varphi^{(k)}(-\mathfrak S,\mathfrak S')
\prod_{(q,j)\in\boldsymbol{\mathsf q}} \delta \varphi^{(j)}(\mathfrak S',-\mathfrak S)
\prod_{(p,k)\in\boldsymbol{\mathsf p}} \delta\bar\varphi^{(k)}(\mathfrak S',-\mathfrak S),
\end{gather*}
where
\begin{gather}\label{dphi}
\delta \varphi^{(k)}(-\mathfrak S,\mathfrak S')=\big(\delta \bar\varphi^{(k)}(-\mathfrak S,\mathfrak S') \big)^{-1}=
\prod_{j=1}^{k-1} \big({\sigma'}^{(k)}-{\sigma'}^{(j)} \big)\prod_{j=1}^N \big(\sigma^{(j)}-{\sigma'}^{(k)} \big).
\end{gather}
Since these expressions are independent of momentum~$p$, it follows that
\begin{gather}\label{zpqzywred}
Z^{\boldsymbol{\mathsf p}',\boldsymbol{\mathsf q}'}_{\boldsymbol{\mathsf p},\boldsymbol{\mathsf q}}
(\mathfrak S', \mathfrak S) = {\tilde{Z}}^{\bs Y', \bs w'}_{\bs Y, \bs w} (\mathfrak S', \mathfrak S)\cdot
\prod_{k=1}^N
\big(\delta \varphi^{(k)}(-\mathfrak S,\mathfrak S')\big) ^{{w'}^{(k)}}
\big(\delta \varphi^{(k)}(\mathfrak S',-\mathfrak S) \big) ^{-{w}^{(k)}}.
\end{gather}

Using (\ref{chY}) we can rewrite
${\tilde{Z}}^{\bs Y', \bs w'}_{\bs Y, \bs w} (\mathfrak S', \mathfrak S)$ in terms of non-charged Young diagrams:
\begin{gather}
{\tilde{Z}}^{\bs Y', \bs w'}_{\bs Y, \bs w} (\mathfrak S', \mathfrak S) =
{\tilde{Z}}^{\bs Y'}_{\bs Y} (\bs \sigma'+ \bs w', \bs\sigma+ \bs w )
\frac{\cN(\bs\sigma+\bs w, 0,\bs\sigma'+\bs w')}{\cN(\bs\sigma,0,\bs\sigma')}\nonumber \\
\hphantom{{\tilde{Z}}^{\bs Y', \bs w'}_{\bs Y, \bs w} (\mathfrak S', \mathfrak S) =}{} \times \prod_{k=1}^N
\big(\delta \Gamma^{(k)}(-\mathfrak S,\mathfrak S')\big) ^{{w'}^{(k)}}
\big(\delta \Gamma^{(k)}(\mathfrak S',-\mathfrak S)\big)^{-{w}^{(k)}}.\label{Zhhat}
\end{gather}
Here the structure constants $\cN(\cdot)$ are defined in \eqref{Nnormcoef}, ${\tilde{Z}}^{\bs Y'}_{\bs Y} (\bs \sigma', \bs\sigma)$ is given by
\begin{gather}
{\tilde{Z}}^{\bs Y'}_{\bs Y} (\bs \sigma', \bs\sigma )=
\prod_{j=1}^N \bigl|Z_{\mathsf{bif}}\big(0\,|\,Y^{(j)},Y^{(j)}\big)\bigr|^{-\frac12}
\bigl|Z_{\mathsf{bif}}\big(0\,|\,{Y'}^{(j)},{Y'}^{(j)}\big)\bigr|^{-\frac12} \nonumber\\
\hphantom{{\tilde{Z}}^{\bs Y'}_{\bs Y} (\bs \sigma', \bs\sigma )=}{}
\times \frac{\prod\limits_{j,k=1}^N Z_{\mathsf{bif}}\big({\sigma'}^{(j)}-\sigma^{(k)}\,|\,{Y'}^{(j)},Y^{(k)}\big)}
{ \prod\limits_{j<k} Z_{\mathsf{bif}}\bigl({\sigma'}^{(j)}-{\sigma'}^{(k)} \,|\,
{Y'}^{(j)},{Y'}^{(k)}\bigr) { Z_{\mathsf{bif}}\big({\sigma}^{(k)}-{\sigma}^{(j)}\,|\, Y^{(k)},Y^{(j)}\big)}},
\end{gather}
and
\begin{gather}\label{dGamma}
\delta \Gamma^{(k)}(-\mathfrak S,\mathfrak S') =
\frac{\prod\limits_{j>k} \Gamma \big(1+{\sigma'}^{(k)}-{\sigma'}^{(j)} \big)}
{\prod\limits_{j<k} \Gamma \big(1+{\sigma'}^{(j)}-{\sigma'}^{(k)}\big) \prod\limits_j\Gamma \big( 1+{\sigma'}^{(k)}-{\sigma}^{(j)}\big)}\,.
\end{gather}

Observe that the products appearing in the right side of \eqref{zpqzywred} and \eqref{Zhhat} have the same structure as the product of $\bar r_k^{[l]}$ and $r_k^{[l]}$ in~\eqref{detzpq}. Combining these products together produces the Fourier exponents in~\eqref{FFF}. Indeed, collect the contributions of the form $\big(B_l^{(k)}\big)^{w_l^{(k)}}$ associated with the circle $\mathcal{C}_l$ coming from the different factors in~(\ref{detAZZZ}):
\begin{gather*} 
Z^{\boldsymbol{\mathsf p}_{l},\boldsymbol{\mathsf q}_{l}}_{\boldsymbol{\mathsf p}_{l+1},\boldsymbol{\mathsf q}_{l+1}}\big(\mathcal T^{[l+1]}\big)
Z^{\boldsymbol{\mathsf p}_{l-1},\boldsymbol{\mathsf q}_{l-1}}_{\boldsymbol{\mathsf p}_{l},\boldsymbol{\mathsf q}_{l}}\big(\mathcal T^{[l]}\big) \quad \Longrightarrow \quad \prod_{k=1}^N \big(B_l^{(k)}\big)^{w_l^{(k)}}.
\end{gather*}
Using (\ref{detzpq}), (\ref{dphi})--(\ref{dGamma}) and then (\ref{RbRres}), we obtain
\begin{gather*}
B_l^{(k)}=\frac{\bar r_k^{[l+1]}}{r_k^{[l]}}\cdot
\frac{\delta \varphi^{(k)}(-\mathfrak S_{l+1}, \mathfrak S_{l})}
{\delta \varphi^{(k)}(\mathfrak S_{l-1}, - \mathfrak S_{l})}\cdot
\frac{\delta \Gamma^{(k)}(-\mathfrak S_{l+1}, \mathfrak S_{l})}
{\delta \Gamma^{(k)}(\mathfrak S_{l-1}, - \mathfrak S_{l})}={\rm e}^{\beta_l^{(k)}}.
\end{gather*}

The dependence of (\ref{detzpq}) on the positions of semi-degenerate poles is concentrated in the factors
\begin{gather*}
z_l^{\boldsymbol{\sigma}_l\cdot\boldsymbol{w}_l-\boldsymbol{\sigma}_{l-1}\cdot\boldsymbol{w}_{l-1}+\sum\limits_{q\in\boldsymbol{\mathsf q}_{l}} q+\sum\limits_{p\in\boldsymbol{\mathsf p}_{l}} p-\sum\limits_{q\in\boldsymbol{\mathsf q}_{l-1}}q-\sum\limits_{p\in\boldsymbol{\mathsf p}_{l-1}}p }\\
\qquad{}=
z_l^{\frac12\big(({\bs \sigma}_{l}+{\bs w}_{l})^2-{\bs \sigma}_{l}^2-({\bs \sigma}_{l-1}+{\bs w}_{l-1})^2+{\bs \sigma}_{l-1}^2\big)-|\bs Y_{l-1}|+|\bs Y_{l}|},
\end{gather*}
where $|\bs Y|$ denotes the total number of cells in the $N$-tuple $\mathbf Y \in \mathbb Y^N$, cf., e.g., \cite[Fig.~13]{GL16}. Taking into account the product of 3-point tau functions in~(\ref{topodec2}), each summand of~(\ref{detAZZZ}) produces the following $\bs z$-dependent contribution to the expansion of~$\tau(\bs z)$:
\begin{gather*}
\prod_{l=1}^{n-2}
z_l^{\frac12\big(({\bs \sigma}_{l}+{\bs w}_{l})^2-({\bs \sigma}_{l-1}+{\bs w}_{l-1})^2+\Lambda_{l}^2\big)-|\bs Y_{l-1}|+|\bs Y_{l}| }.
\end{gather*}
The remaining factors ${\tilde{Z}}^{\bs Y'}_{\bs Y} (\bs \sigma', \bs\sigma)$, coming to $Z^{\boldsymbol{\mathsf p}',\boldsymbol{\mathsf q}'}_{\boldsymbol{\mathsf p},\boldsymbol{\mathsf q}}(\mathfrak S', \mathfrak S)$ from ${\tilde{Z}}^{\bs Y', \bs w'}_{\bs Y, \bs w} (\mathfrak S', \mathfrak S)$, reproduce the gauge theory prediction \eqref{GTF} for the matrix elements of semi-degenerate vertex operators. To identify the corresponding formulas, one should take into account that Cauchy matrices \eqref{cael1}--\eqref{cael2} arise when $\operatorname{rk}A_1=\dots=\operatorname{rk}A_{n-2} =1$ in~\eqref{ffs}. The transformation to the ``CFT gauge'' (where all $\operatorname{Tr}A_k=0$) shifts $\mathfrak S_k$'s by appropriate scalars, which effectively leads to $\frac \Lambda N$-translations in the formulas such as~\eqref{GTF} and~\eqref{Nnormcoef}.

\subsection*{Acknowledgements}
The work of P.G.\ was partially supported by the Russian Academic Excellence Project \smash{`5-100'} and by the RSF grant No.~16-11-10160. In particular, the results of Section~\ref{sec_fred} were obtained using support of Russian Science Foundation. P.G.\ is a Young Russian Mathematics award winner and would like to thank its sponsors and jury. N.I.\ was partially supported by the National Academy of Sciences of Ukraine (project No.~0117U000238), by the Program of Fundamental Research of the Department of Physics and Astronomy of the NAS of Ukraine (project No.~0117U000240), and by the ICTP-SEENET-MTP project NT-03: Cosmology -- Classical and Quantum Challenges.

\pdfbookmark[1]{References}{ref}
\LastPageEnding

\end{document}